\documentclass{zippostyle}

\newcommand{\myparagraph}[1]{\noindent{\textbf{#1}. }}

\def\primecon{\mathfrak P}
\def\primesign{\mathcal P}

\def\aryof#1{\sizeof{#1}}

\def\form{formula\xspace}
\def\forms{formulas\xspace}

\def\mllform{$\MLL$-formula\xspace}

\newcommand*{\con}[1][]{\kappa_{#1}}
\newcommand*{\ncon}[1][]{\cneg\kappa_{#1}}
\def\idperm{\mathsf{id}}
\def\permset#1{\mathfrak{S}_{#1}}
\def\pomset{\mathsf{Pomset}}
\def\BV{\mathsf{BV}}

\def\GVsl{\mathsf{GV^{sl}}}

\def\PML{\mathsf{MGL}}
\def\PMLx{\mathsf{MGL}^\funit}

\def\PLK{\mathsf{GLK}}

\def\GBL{\mathsf{GBL}}

\def\GMLx{\mathsf{GML}^\funit}

\mathchardef\mhyphen="2D

\def\Sys{\mathsf S}
\def\RB{\textsf{RB}}

\newvertex{cdots}{\cdots}{}
\newvertex{vdots}{\vdots}{}

\def\NP{\mathbf{NP}}

\def\transto{\rightsquigarrow}

\def\set#1{\{#1\}}
\def\Set#1{\left\{#1\right\}}

\def\tuple#1{\langle#1\rangle}
\usepackage{scalerel}
\def\Tuple#1{\stretchleftright[1000]{\langle}{#1}{\rangle}}
\def\intset#1#2{\set{#1,\ldots,#2}}
\newcommand{\sizeof}[1]{|#1|}

\renewcommand{\emptyset}{\varnothing}

\def\RBPN{\RB-proof net\xspace}

\newcommand{\atomset}{\mathcal A}
\newcommand{\litset}{\mathcal L}

\def\cneg#1{#1^\lbot}
\def\cnegneg#1{#1^{\lbot\lbot}}

\def\unit{\emptyset}

\newcommand{\mlpar}[1][n]{\lpar_{#1}}

\def\funit{\circ}
\def\funits{\funit,\ldots,\funit}

\def\fcon#1{\kappa_{#1}}
\def\fP{\fcon{\gP}}
\def\nfP{\fcon{\nP}}

\def\fA{\phi}
\def\nfA{\cneg\fA}
\def\fB{\psi}
\def\nfB{\cneg\psi}
\def\fC{\chi}
\def\nfC{\cneg\chi}
\def\fM{\mu}
\def\nfM{\cneg\mu}
\def\fN{\nu}
\def\nfN{\cneg\nu}
\def\fcC{\zeta}

\def\fx{x}

\def\Sym#1{\mathfrak S(#1)}
\def\Dsym#1{\mathfrak S^\lbot(#1)}

\newcommand{\connn}[1]{\llparenthesis#1\rrparenthesis}
\newcommand{\Connn}[1]{\begin{pmatrix}\!\begin{vmatrix}#1\end{vmatrix}\!\end{pmatrix}}
\newcommand{\cons}[1]{[#1]}
\newcommand{\Cons}[1]{\left[#1\right]}
\newcommand{\chole}{\lbox}
\newcommand{\conso}{\cons{\chole}}
\newcommand{\consu}{\cons{\unit}}
\newcommand{\consfu}{\cons{\funit}}

\def\defn#1{{\itshape\bfseries\boldmath #1}}

\newcommand{\quand}{\quad\mbox{and}\quad}
\newcommand{\qquand}{\qquad\mbox{and}\qquad}

\newcommand{\quor}{\quad\mbox{or}\quad}
\newcommand{\qquor}{\qquad\mbox{or}\qquad}

\newcommand{\proves}[1][]{\mathord{\vdash_{#1}\,}}
\def\dD{\mathcal D}

\def\comp{\; ; \;}

\newcommand{\MLL}{\mathsf{MLL}}
\newcommand{\MLLx}{\MLL^{\funit}}
\def\GS{\mathsf{GS}}

\newcommand{\XS}[1][]{\mathsf{X}_{#1}}
\def\sysS{\mathsf S}

\newcommand{\unitor}[1][]{\mathsf{unitor}_{\con}}
\def\axrule{\mathsf {ax}}
\def\AXrule{\mathsf {AX}}
\def\cutr{\mathsf {cut}}
\def\mixr{\mathsf{mix}}
\def\airule{\mathsf{ai}}

\def\swir{\mathsf{s}}
\def\prule{\mathsf{p}}
\def\rrule{\mathsf{r}}

\def\crule{\mathsf{c}}
\def\acrule{\mathsf{ac}}

\def\wrule{\mathsf w}
\def\crule{\mathsf c}
\def\medr{\mathsf m}

\def\conrule{\mathsf{cxt}}
\def\deepr{\mathsf{deep}}
\def\wdtr{\mathsf{wd}_{\ltens}}
\def\wdpr{\mathsf{wd}_{\lpar}}
\def\cpr{\conrule\mhyphen{\lpar}}

\def\aidr{\airule\mathord{\downarrow}}

\def\sdr{\swir_{\lpar}}
\def\sur{\swir_{\ltens}}

\def\pdr{\prule\mathord{\downarrow}}

\def\dwrule{\wrule\mathord{\downarrow}}
\def\dcrule{\crule\mathord{\downarrow}}
\def\dacrule{\acrule\mathord{\downarrow}}

\def\single{\textsf{s}}
\def\double{\textsf{d}}
\newcommand{\conr}[1][\con]{\single\mhyphen#1}
\newcommand{\dconr}[1][\con]{\double\mhyphen#1}

\newcommand{\symr}[1][\con]{\mathsf{sym}\mhyphen#1}
\newcommand{\dsymr}[1][\con]{\mathsf{dsym}\mhyphen#1}
\newcommand{\tassr}[1]{\ltens\mhyphen\mathsf{asso}}
\newcommand{\passr}[1]{\lpar\mhyphen\mathsf{asso}}

\def\RBrule#1{#1^{\RB}}
\newcommand{\rbconr}[1][\gQ]{\single\mhyphen\RBrule{\expandafter\con[#1]}}
\newcommand{\rbdconr}[1][\gQ]{\double\mhyphen\RBrule{\expandafter\con[#1]}}
\newcommand{\rbGdconr}[1][\gG]{\double\mhyphen\RBrule{\expandafter\con[#1]}}

\def\ltensr{\ltens}
\def\lparr{\lpar}

\def\IH{\mathsf{IH}}

\newcommand{\vertices}[1][]{V_{#1}}

\newcommand{\labelset}{\mathcal L}
\def\labsymb{\ell}

\newcommand{\lab}[2][]{\labsymb_{#1}(#2)}
\newcommand{\vof}[1]{\vertices[#1]}

\def\emptygraph{\emptyset}
\def\restr#1{|_{#1}}

\def\greq{=}
\newcommand{\isym}[1][]{\sim_{#1}}

\newcommand{\gof}[1]{\left[\!\left[ #1 \right]\!\right]}
\newcommand{\fof}[1]{\left[\!\left[ #1 \right]\!\right]^{-1}}

\def\gC{C}

\def\gF{F}
\def\gG{G}
\def\gH{H}

\def\gM{M}
\def\gN{N}
\def\gP{P}
\def\gQ{Q}

\def\nP{\cneg\gP}

\def\nG{\cneg\gG}
\def\nH{\cneg\gH}

\def\gPath#1{\mathsf{P_{#1}}}

\def\Pfour{\mathsf{P_4}}

\def\nPfour{\mathsf{\cneg P_4}}

\def\Pbull{\mathsf{Bull}}

\def\clique#1{\mathsf{K}_#1}

\def\stable#1{\mathsf{S}_#1}

\def\vG{\vof\gG}

\def\eG{\eof\gG}
\def\vP{\vof\gP}

\def\nG{\cneg\gG}

\def\nP{\cneg\gP}

\def\vC{\vof\cC}
\def\uC{\uedge[\cC]}

\def\uG{\uedge[\gG]}

\def\eG{E_\gG}

\def\lG{\labsymb_\gG}

\def\nuG{\nuedge[\gG]}
\def\lnG{\labsymb_{\cneg\gG}}

\def\cC{\mathcal C}

\newcommand{\uedge}[1][]{\mkern1mu\mathord{\stackrel{#1}{{\color{black}\frown}}}\mkern1mu}
\newcommand{\nuedge}[1][]{\mkern1mu\mathord{\stackrel{#1}{\not\frown}}\mkern1mu}

\def\odiso#1#2{\od{\odo{\odh{#1}}{}{#2}{}}}

\def\ods#1#2#3#4{\od{\odd{\odh{#1}}{#2}{#3}{#4}}}
\newcommand{\vldr}[2]{\vlpr{#1}{}{#2}}

\defedgetype{DOT}{>=stealth,->,draw,densely dotted}{}
\defedgetype{D}{>=stealth,->,draw,pzgreen}{}
\defedgetype{dR}{draw=cographcolor,densely dotted}{}
\defedgetype{RB}{very thick,dashed}{}

\newvertex{bul}{\bullet}{}
\newvertex{un}{\unit}{}
\newvertex{gdots}{\cdots}{}

\newemptyvertex{nmod}{draw,rounded corners=5,inner sep=2pt,minimum size=10pt,fill=red,opacity=.5,text opacity=1}
\newemptyvertex{emod}{}

\def\vind#1{
	\mathord{%
		\tikz[remember picture,anchor=base,baseline]%
		\node[inner sep=0pt](v#1){$v_{#1}\strut$};
	}
}

\def\vvo#1#2{
	\mathord{%
		\tikz[remember picture,anchor=base,baseline]%
		\node[inner sep=0pt](vo#2){$\mathsf o_{#1}\strut$};
	}
}

\def\vvr#1#2{
	\mathord{%
		\tikz[remember picture,anchor=base,baseline]%
		\node[inner sep=0pt](vr#2){$\mathsf r_{#1}\strut$};
	}
}

\def\sqBedge#1#2#3{
	\tikz[overlay,remember picture,draw,fill,opacity=1] %
	\draw[draw=linkcolor,very thick] (#1) -- ++(0,#3pt) -|  (#2);
}

\def\sqBedges#1{
	\foreach \aaa/\bbb/\ccc in {#1} {\sqBedge{\aaa}{\bbb}{\ccc}}%
}

\def\Hpath#1{
	\tikz[overlay,remember picture,draw,fill,opacity=1] %
	\draw[draw=yellow,line width = 4pt,opacity=.5] #1;
}

\newvertex{pa}{a'}{}
\newvertex{pb}{b'}{}
\newvertex{pc}{c'}{}
\newvertex{pd}{d'}{}

\def\feq{\equiv}

\title{
	Graphical Proof Theory I:
	\\[10pt]
	Sequent Systems on Undirected Graphs
}
\author{Matteo Acclavio}

\begin{document}

\maketitle
\begin{abstract}	
	In this paper we explore the design of sequent calculi operating on graphs. For this purpose, we introduce a set of logical connectives extending the well-known correspondence between classical propositional formulas and cographs, and we define sequent systems operating on formulas over these connectives.
	
	We prove, using an analyticity argument based on cut-elimination, that our systems provide conservative extensions of multiplicative linear logic (without and with mix) and classical propositional logics. We conclude by showing that one of our systems captures graph isomorphism as logical equivalence, and that this system is also sound and complete for the graphical logic $\GS$.
\end{abstract}
\tableofcontents
\clearpage

\section{Introduction}

In theoretical computer science, \emph{formulas} play a crucial role in describing complex abstract objects. 
At the syntactical level, the formulas of a logic describe complex structures by means of unary and binary operators, usually thought of as \emph{connectives} and \emph{modalities} respectively.
On the other hand, graph-based syntaxes are often favored in formal representation, as they provide an intuitive and canonical description of properties, relations and systems.  
By means of example, consider the two graphs below:
$$
\va1\qquad \vb1 \qquad \vc1 \qquad \vd1
\Dedges{b1/a1,b1/c1,d1/c1}
\qquad\quor\qquad
\va1\qquad \vb1 \qquad \vc1 \qquad \vd1
\Gedges{b1/a1,b1/c1,d1/c1}
$$
It follows from results in \cite{Valdes1979,cographs} that describing any of the above graphs by means of formulas only employing binary connectives would require repeating at least one vertex.
As a consequence, formulas describing complex graphs are usually long and convoluted, and a specific \emph{encodings} are needed to standardize such formulas. 

Since graphs are ubiquitous in theoretical computer science and its applications, 
a natural question to ask is whether it is possible to define formalisms having graphs, instead of formulas, as first-class terms of the syntax.
Such a paradigm shift would allow to design efficient automated tools free from the bureaucracy introduced to handle the encoding required to represent graphs. 
At the same time, a graphical syntax would provide a useful tool for investigations such as the ones in \cite{learningPomset} or \cite{Fu2004,Denielou2010}, where the authors restrain their framework to sequential-parallel orders, as these can be represented by means of formulas with at most binary connectives.

Two recent lines of works have generalized proof theoretical methodologies to graphs, extending the correspondence between classical propositional formulas and cographs.
In these works, systems operating on graphs are defined via local and context-free rewriting rules, similarly as what done in \emph{deep inference} systems \cite{gug:SIS,gug:gun:par:2010,tub:str:esslli19}.
The first line of research, carried out by Calk, Das, Rice and Waring in various works \cite{CDW:ext-bool,calk:graph,waring:master,das:19,das:rice:FSCD2021}, explores the use of maximal stable sets/cliques-preserving homomorphisms to define notions of entailment\footnote{
	A similar approach was proposed in \cite{pratt1986modeling} for studying pomsets.
}, and study the resulting proof theory.
Here, the choice of the using of a deep inference formalism is natural, since the rules of the calculus are local rewritings.
The second line of research, investigated by the author, Horne, Mauw and Stra\ss burger in several contributions~\cite{acc:hor:str:LICS2020,acc:LMCS,acc:FSCD22}, studies the (sub-)structural proof theory of arbitrary graphs, with an approach inspired by linear logic~\cite{girard:87} and deep inference~\cite{gug:SIS}. 
The main goal of this line of research, partially achieved with the system $\GVsl$ operating on mixed graphs \cite{acc:FSCD22}, is to obtain a generalization of the completeness result of the logic $\BV$ with respect to pomset inclusion.
The logic $\BV$ contains a non-commutative binary connective $\lseq$ allowing to represent series-parallel partial order multisets as formulas in the syntax (as in Retoré's $\pomset$ logic \cite{ret:newPomset}), and to capture order inclusion as logical implication. However, as shown in \cite{tiu:SIS-II}, no cut-free sequent system for $\BV$ can exists -- therefore neither for $\pomset$ logic, which strictly contains it \cite{tito:lutz:csl22,tito:str:SIS-III}. 
For this reason the aforementioned line of work focused on deep inference systems, and the question about the existence of a cut-free sequent calculus for $\GS$ (the restriction of $\GVsl$ on undirected graphs originally defined in \cite{acc:hor:str:LICS2020}) was left open.

\paragraph{Main contributions}

In this paper we focus on the definition of sequent calculi for \emph{graphical logics}, and we positively answer the above question by providing, among other results, a cut-free sound and complete sequent calculus for $\GS$. 
By using standard techniques in sequent calculus, we thus obtain a proof of analiticity for this logic which is simpler and more concise with respect to the one in \cite{acc:LMCS}. 

To achieve these results, we introduce \emph{graphical connectives}, which are operators that can be naturally interpreted as graphs. 
We then define the sequent calculi $\PML$, $\PMLx$ and $\PLK$, containing rules to handle these connectives.
After showing that cut-elimination holds for these systems, 
we prove that  $\PML$, $\PMLx$ and $\PLK$ define conservative extensions of \emph{multiplicative linear logic}, \emph{multiplicative linear logic with mix} and \emph{classical propositional logic} respectively.
We then prove that formulas interpreted as the same graph are logically equivalent, thus justifying the fact that we consider these systems as operating on graphs rather than formulas.
We conclude by showing that $\PMLx$ is sound and complete with respect to the logic $\GS$, thus providing a simple sequent calculus for the logic.%

\paragraph{Outline of the paper}
In \Cref{sec:graphs} we recall definitions and results in graph theory and the notion of modular decomposition. 
In \Cref{sec:calculi} we use these notions to extend the correspondence between classical propositional formulas and cographs to any graph.
We define linear sequent calculi and we prove their properties.
In \Cref{sec:GS} we show that one of these calculi is sound and complete with respect to the set of non-empty graphs provable in the deep inference system $\GS$ studied in \cite{acc:hor:str:LICS2020,acc:LMCS}.
In \Cref{sec:LK} we define a proof system which is a conservative extension of classical logic.
To conclude, we summarize in \Cref{sec:conc} some of the possible the research directions opened by this work.

%
%
%

\section{From Formulas To Graphs}\label{sec:graphs}

In this section we recall standard results from the literature on graphs such as \emph{modular decomposition} and \emph{cographs}.
We then introduce the notion of \emph{graphical connectives}
allowing us to extend the correspondence between cographs and classical propositional formulas to general graphs.

\subsection{Graphs and Modular Decomposition}

In this work are interested in using graphs to represent patterns of interactions by means of the binary relations (edges) between their components (vertices).
For this reason we recall the definition of 
\emph{labeled graph} (the mathematical structure we use to encode these patterns)
together with 
the definition of 
\emph{isomorphism} (the standard notion of identity on labeled graphs)
and
the rougher notion of \emph{similarity} (equivalence up-to labels over vertices).

\begin{definition}\label{def:graph}
	A \defn{$\litset$-labeled graph} (or simply \defn{graph})
	$\gG=\tuple{\vG,\lG,\uG}$
	is given by 
	a finite set of \defn{vertices} $\vG$,
	a partial \defn{labeling function} 
	$\labsymb_\gG\colon \vertices[\gG]\to \litset$ 
	associating a label $\lab v$ from a given set of labels $\litset$ to each vertex $v\in\vG$
	(we may represent $\labsymb_\gG$ as a set of equations of the form $\lab v= \ell_v $ and denote by $\emptyset$ the empty function),
	and
	a non-reflexive symmetric edge relation $\uG\subset\vof\gG\times\vof\gG$ whose elements, called \defn{edges}, may be denoted $vw$ instead of $(v,w)$.
	The \defn{empty} graph $\tuple{\emptyset,\emptyset,\emptyset}$ is denoted $\emptygraph$.
	
	A \defn{similarity} between two graphs $\gG$ and $\gG'$
	is a bijection 
	$f\colon \vG \to \vof{\gG'}$
	such that 
	$x\uedge[\gG] y$ iff $f(x)\uedge[\gG'] f(y) $
	for any $x,y\in\vG$.
	An  \defn{isomorphism} is a similarity $f$ such that $\lab{v}=\lab{f(v)}$ for any $x,y\in\vG$.
	Two graphs
	$\gG$ and $\gG'$ are 
	\defn{similar} (denoted $\gG\isym\gG'$)
	if there is an similarity between $\gG$ and $\gG'$.
	A \defn{symmetry} is a similarity of a graph with itself.
	They are 
	\defn{isomorphic} (denoted $\gG=\gG'$)
	if there is a isomorphism between $\gG$ and $\gG'$.
	From now on, we consider two isomorphic graphs to be \defn{the same} graph.

	Two vertices $v$ and $w$ in $\gG$ are \defn{connected}
	if there is a sequence  $v=u_0,\ldots, u_n=w$  of vertices in $\gG$ (called \defn{path}) 
	such that $u_{i-1}\uedge[\gG]u_{i}$ for all $i\in\intset1n$.
	A \defn{connected component} of $\gG$ is a maximal set of connected vertices in $\gG$.
	A graph $\gG$ is a \defn{clique} (resp.~a \defn{stable set})
	iff
	$\nuedge[\gG]=\emptyset$ (resp.~$\uG=\emptyset$).
\end{definition}

\begin{nota}\label{nota:mod}
	When drawing a graph or an unlabeled graph
	we draw $\vv1\quad\vw1\Gedges{v1/w1}$ whenever $v\uedge w$, 
	we draw no edge at all whenever $v\nuedge w$. 
	We may represent a vertex of a graph by using its label instead of its name. 
	For example, the single-vertex graph $\gG=\tuple{\set{v},\lG,\emptyset}$ may be represented either by a the vertex name $\vv1$ or by the vertex label $\lab v$ (or $\bullet$ if $\lab v $ is not defined).
	Note that,
	since we are considering isomorphic graphs to be the same, 
	as soon as there is no ambiguity due to vertices represented by the same symbol,
	we can assume
	that
	the representation of a graph to provide us one of the possible triple (set of vertices, label function, and set of edges) defining it.
\end{nota}

\begin{example}\label{ex:con1}
	Consider the following graphs:
	$$\begin{array}{c@{\;=\;}l}
		\gF	&	\Tuple{\Set{u_1,u_2,u_3,u_4},\Set{\lab{u_1}=a,\lab{u_2}=b,\lab{u_3}=c,\lab{u_4}=d},\Set{u_1u_2,u_2u_3,u_3u_4}}
		\\
		\gG	&	\Tuple{\Set{v_1,v_2,v_3,v_4},\Set{\lab{v_1}=b,\lab{v_2}=a,\lab{v_3}=c,\lab{v_4}=d},\Set{v_1v_2,v_1v_3,v_3v_4}}
		\\
		\gH	&	\Tuple{\Set{w_1,w_2,w_3,w_4},\Set{\lab{w_1}=a,\lab{w_2}=b,\lab{w_3}=c,\lab{w_4}=d},\Set{w_1w_2,w_1w_3,w_3w_4}}
	\end{array}
	$$
	They are all symmetric, that is $\gF\isym\gG\isym\gH$, but $\gF=\gG\neq \gH$ as can easily be verified using their representations:
	$$
	\gF=\va1\quad\vb1\quad\vc1\quad\vd1\Gedges{a1/b1,b1/c1,c1/d1}=\gG
	\qquand
	\gH=\vb1\quad\va1\quad\vc1\quad\vd1\Gedges{b1/a1,a1/c1,c1/d1} 
	$$
\end{example}

\begin{obs}
	The problem of graph isomorphism is a standard $\NP$-problem (to be more precise, its complexity is quasi-polynomial \cite{bab:quasi}).
	That is, verify that a given bijection between the sets of vertices of two graphs is an isomorphism can be checked in polynomial time,
	while there is no known polynomial time algorithm to find such an isomorphism.
	For this reason, whenever we say that two graphs are the same, either we assume they share the same set of vertices, therefore implicitly assuming the isomorphism $f$ to be defined by the identity function over the set of vertices, or we assume an isomorphism to be given.
	This allows us to verify whether two graphs are the same in polynomial time.
\end{obs}

In order to use proof theoretical methodologies on graphs, we need a suitable notion of subgraphs to be used in the same way sub-formulas are used in proof systems, that is, to state properties of the calculus or to define the behavior of rules.
For this purpose, we use for a notion of  \emph{module} to identify subgraph allowing us to decompose a graph using abstract syntax trees similar to the ones underlying formulas \cite{gallai:67,james1972graph,hab:paul:survey,lovasz2009matching,mcc:ros:spi:linear,Ehrenfeucht1999}.
A module is a subset of vertices of a graph having the same edge-relation with any vertex outside the subset. 
This definition generalizes the interaction we usually be observed in formulas, where, in the formula tree, any literal in a subformula has the same relation (the one given by the least common ancestor) with a given literal not occurring in the subformula itself.

\begin{definition}
	Let $\gG=\tuple{\vG,\lG,\eG}$ be a graph and $W\subseteq\vG$.
	The \defn{graph induced} by $W$ is
	the graph 
	$\gG\restr{W}\coloneqq\Tuple{W,\lG\restr{W},\uG\cap\left(W\times W\right)}$
	where
	$\lG\restr{W}(v)\coloneqq \lG(v)$ for all $v\in W$.
	
	A \defn{module} of a graph $\gG$ is a subset $M$ of $\vG$ such that
	$x\uedge z$ iff $ y\uedge z $ 
	for any $x,y\in M$, 
	$z\in \vG \setminus M$.
	A module $M$ is \defn{trivial} if $M=\emptyset$,
	$M=\vG$, or $M=\set{x}$ for some $x\in\vG$.
	From now on, we identify a module $\gM$ of a graph $\gG$ with 
	the induced subgraph $\gG\restr{M}$.
\end{definition}
\begin{remark}
	A connected component of a graph $\gG$ is a module of $\gG$.
\end{remark}

Using modules we can optimize the way we represent graphs reducing the number of edges drawn without losing information, relying on the fact that all vertices of a module has the same edge-relation with any vertex outside the module.

\begin{nota}
	In representing graphs we may border vertices of a same module by a closed line.
	An edges connected to such a closed line denotes the existence of an edge to each vertex inside it.
	By means of example, consider the following graph and its more compact modular representation.
	
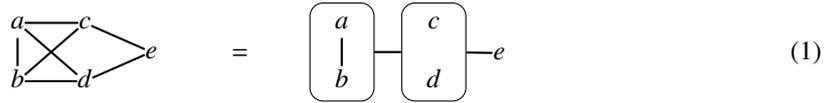
\begin{equation}
		\begin{array}{c@{\qquad}c@{\qquad}c}
			\va1&\vc1
			\\
			&&\ve1
			\\
			\vb1&\vd1
		\end{array}
		\Gedges{a1/b1,a1/c1,a1/d1,b1/d1,b1/c1,e1/c1,e1/d1}
		\qquad=\qquad
		\vmod1{\begin{array}{c}\va1\\\\\vb1\end{array}}\quad\vmod2{\begin{array}{c}\vc1\\\\\vd1\end{array}}\quad \ve1
		\Gedges{a1/b1,mod1/mod2,mod2/e1}
	\end{equation}
\end{nota}

The notion of module is related to a notion of context, which can be intuitively formulated as a graph with a special vertex playing the role of a hole in which we can plug in a module.
\begin{definition}
	A \defn{context} $\cC\conso$ is a (non-empty) graph containing a single occurrence of a special vertex $\chole$ (such that $\lab\chole$ is undefined).
	It is \defn{trivial} if $\cC\conso=\chole$.
	If $\cC\conso$ is a context and $\gG$ a graph, we define $\cC\cons{\gG}$ as the graph obtained by replacing $\chole$ by $\gG$.
	Formally,
	\begin{equation*}
		\cC\cons\gG
		\coloneqq
		\Tuple{	
			\left(\vertices[\cC\conso]\setminus\set{\chole}\right)\uplus \vG 
			\;,\;
			\labsymb_\cC\cup\lG
			\;,\;
			\Set{vw\mid v,w\in \vertices[\cC\conso]\setminus\set{\chole}, v\uedge[\cC\conso]w}
			\cup 
			\Set{
				vw
				\mid 
				v\in \vertices[\cC\conso]\setminus\set{\chole}, 
				w\in\vG, 
				v\uedge[\cC\conso]\chole
			}
		}
	\end{equation*}
\end{definition}
\begin{remark}
	A set of vertices $\gM$ is a module of a graph $\gG$ 
	iff 
	there is a context $\cC\conso$ such that $\gG=\cC\cons{\gM}$.
\end{remark}

We generalize this idea of replacing a vertex of a graph with a module
by defining the operations of \emph{composition-via} a graph, where all vertices of a graph are replaced in a ``modular way'' by modules.

\begin{definition}\label{def:compVia}
	Let $\gG$ be a graph with $\vof\gG=\set{v_1,\ldots,v_n}$ 
	and 
	let $\gH_1,\ldots,\gH_n$ be graphs.
	We define the
	\defn{composition of $\gH_1,\ldots,\gH_n$ via $\gG$}
	as the graph $\gG\connn{\gH_1,\ldots,\gH_n}$
	obtained by replacing each vertex $v_i$ of $\gG$
	with a module $\gH_i$ for all $i\in\intset1n$.
	Formally,
	\begin{equation}\label{eq:compVia}
		\gG\connn{\gH_1,\ldots,\gH_n}
		\;=\;
		\Tuple{
			\;
			\biguplus_{i=1}^n\vof{\gH_i} 
			\;,\;
			\bigcup_{i=1}^n{\labsymb_{\gH_i}}
			\;,\;
			\left(\bigcup_{i=1}^n\uedge[\gH_i]\right)
			\cup
			\Set{\!\!
				\begin{array}{c|c}
					(x,y)
					&
					x\in \vof{\gH_i},
					y\in\vof{\gH_j},
					v_i\uedge[\gG] v_j 
				\end{array}\!\!
			}
			\;
		}
	\end{equation}	
	The subgraphs $\gH_1,\ldots, \gH_n$ are called \defn{factors} of $\gG\connn{\gH_1,\ldots, \gH_n}$
	and, by definition, 
	are (possibly not maximal) modules of $\gG\connn{\gH_1,\ldots, \gH_n}$.
\end{definition}
\begin{remark}\label{rem:compVia}
	The information about the labels of the graph $\gG$ used to define the composition-via operation is lost.
	Moreover, 
	if $\gG$ is a graph with $\vG=\set{v_1,\ldots, v_n}$ and 
	$\sigma$ a permutation over the set $\intset1n$ 
	such that 
	the map $f_\sigma:\vG\to\vG$ mapping $v_i$ in $f_\sigma(v_i)=v_{\sigma(i)}$ for all $i\in\intset1n$ is an similarity between $\gG$ and $\gG$, 
	then $\gG\connn{\gH_1,\ldots,\gH_n}=\gG'\connn{\gH_1,\ldots,\gH_n}$.
\end{remark}

In order to establish a connection between graphs and formulas, from now on we only consider graphs whose set of labels belong to the set 
$\labelset=\Set{a,\cneg a\mid a\in\atomset}$ where $\atomset$ is a fixed set of propositional variables.
We then define the \emph{dual} of a graphs.

\begin{definition}
	Let $\gG=\tuple{\vG,\lG,\eG}$ be a graph.
	We define the edge relation $\nuG\coloneqq\Set{(v,w)\mid v\neq w \mbox{ and } vw\notin\uG}$
	and
	we define the \defn{dual} graph of $\gG$ as the graph 
	$\nG\coloneqq \Tuple{\vG,\nuedge[\gG],\lnG}$ with $\lnG(v)=\cneg{\left(\lG(v)\right)}$ (assuming $\cnegneg a=a$ for all $a\in\atomset$).
\end{definition}
\begin{remark}
	By definition, each module of a graph corresponds to a module of its dual graph.
	It follows that a connected component of $\nG$ is a module of $\gG$.
\end{remark}
\begin{nota}
	If $\mathcal G$ is the representation of a graph $\gG$, 
	then we may represent the graph $\nG$  by bordering the representation of $\gG$ with a closed line with the negation symbol on the upper-right corner, that is,  $\cneg{\vmod1{\mathcal G}}$.
\end{nota}

\subsection{Classical Propositional Formulas and Cographs}\label{sec:cographs}

The set of \defn{classical (propositional) formulas} is generated from a set of propositional variable $\atomset$  using the \defn{negation} $\cneg{(\cdot)}$, the \defn{disjunction} $\lor$ and the \defn{conjunction} $\land$
using the following grammar:
\begin{equation}\label{eq:LKform}
	\fA,\fB			\coloneqq 
	a 				\mid  
	\fA \lor \fB 	\mid 
	\fA \land \fB 	\mid
	\cneg\fA
	\qquad
	\mbox{with $a\in\atomset$.}
\end{equation}
We denote by $\feq$  the equivalence relation over formulas generated by the following laws:
\begin{equation}\label{eq:formEQ}
	\begin{array}{rl}
		\mbox{\defn{Equivalence laws}}
		&
		\left\{\begin{array}{r@{\;\feq\;}l@{\hskip4em}r@{\;\feq\;}l}
			\fA \lor \fB			&\fB\lor \fA
			&
			\fA \lor (\fB\lor \fC)	&(\fA\lor \fB)\lor \fC
			\\
			\fA \land \fB			&\fB\land \fA
			&
			\fA \land (\fB\land \fC)&(\fA\land \fB)\land \fC
		\end{array}\right.
		\\[8pt]
		\mbox{\defn{De-Morgan laws}}
		&
		\;\left\{
		\begin{array}{r@{\;\feq\;}l@{\hskip6.8em}r@{\;\feq\;}l}
			\!\cneg{(\cneg \fA)}		&\fA
			&
			\cneg{(\fA\land\fB)}	&\cneg \fA \lor \cneg\fB
		\end{array}
		\right.
	\end{array}
\end{equation}	
%
We define a map from literals to single-vertex graphs, which extends to formulas via the composition-via 
the unlabeled 
two-vertices stable set $\stable2$
and
two-vertices clique $\clique2$.
%
\begin{definition}
	Let $\fA$ be a classical formula, then $\gof \fA$ is the graph inductively defined as follows:
	$$
	\gof a = a
	\qquad
	\gof {\cneg \fA} = \cneg{\gof\fA}
	\qquad
	\gof {\fA\lor\fB} = \stable2\Connn{\gof\fA , \gof\fB}
	\qquad
	\gof {\fA\land\fB} = \clique2\Connn{\gof\fA , \gof\fB}
	$$
	where $\stable2$ and $\clique{2}$ are respectively a stable set and a clique with $2$ vertices,
	and
	where we denote by $a$ the single-vertex graph, whose vertex is labeled by $a$.
\end{definition}

We can easily observe that the map $\gof\cdot$ well-behaves with respect to the equivalence over formulas $\feq$, that is, equivalent formulas are mapped to the symmetric graphs.
\begin{proposition}
	Let $\fA$ and $\fB$ be classical formulas. 
	Then $\fA\feq\fB$ iff $\gof\fA\greq\gof\fB$.
\end{proposition}

We finally recall the definition of \emph{cographs}, and the theorem establishing the relation between cographs and classical formulas, i.e., providing an alternative definition of cographs as graphs generated by single-vertex graphs using the composition-via a two-vertices no-edge graph and a two-vertices one-edge graph.

\begin{definition}
	A \defn{cograph} is a graph $\gG$ such that 
	there are no four vertices $v_1,v_2,v_3,v_4$ in $\gG$
	such that the induced subgraph $\gG\restr{\set{v_1,v_2,v_3,v_4}}$ is similar to the graph $\tuple{\set{a,b,c,d},\emptyset,\set{ab,bc,cd}}=\va1\quad\vb1\quad\vc1\quad\vd1\Gedges{a1/b1,b1/c1,c1/d1}$ .
\end{definition}
\begin{theorem}[\cite{gallai:67}]\label{thm:cograph}
	A graph $\gG$ is a cograph iff there is a formula $\fA$ such that $\gG\isym\gof{\fA}$.
\end{theorem}

\subsection{Modular Decomposition of Graphs}\label{sec:modDec}

We recall the notion of \emph{prime graph}, 
allowing us to provide canonical representatives of graphs via modular decomposition.
(see e.g., \cite{gallai:67,james1972graph,hab:paul:survey,lovasz2009matching,mcc:ros:spi:linear,Ehrenfeucht1999}).
\begin{definition}
	A graph $\gG$ is \defn{prime} if $\sizeof{\vG}> 1$ and all its modules are trivial. 
\end{definition}
We recall the following standard result from the literature.
\begin{theorem}[\cite{james1972graph}]\label{thm:modDec}
	Let $\gG$ be a graph with at least two vertices.
	Then there are non-empty modules $ \gM_1, \dots, \gM_n $ of $\gG$ and a prime graph $\gP$ such that $\gG=\gP\connn{\gM_1, \dots, \gM_n}$.
\end{theorem}

This result implies the possibility of describing graphs using single-vertex graphs and the operation of composition-via prime graphs.
More precisely, we can define the notion of \emph{modular decomposition} of a graph composition-via prime graphs to provide a more canonical representation.
\begin{definition}\label{def:modDec}
	Let $\gG$ be a non-empty graph.
	A \defn{modular decomposition} of $\gG$ is a way to write $\gG$ using single-vertex graphs and the operation of composition-via prime graphs:
	\begin{itemize}
		
		\item if $\gG$ is a graph with a single vertex $x$ labeled by $a$, then $\gG=a$ 
		(i.e., $\gG=\tuple{\set{x},\lab{x}=a,\emptyset}$);
		
		\item if $\gH_1, \ldots, \gH_n$ are maximal modules of $\gG$ such that $\vertices[\gG]=\biguplus_{i=1}^n\vertices[\gH_i]$, 
		then there is a unique prime graph $\gP$ such that $\gG=\gP\connn{\gH_1, \ldots, \gH_n}$.
	\end{itemize}
\end{definition}

\begin{remark}
	There are various reasons why modular decomposition is not unique.
	
	The first is due to the possible presence of cliques and stable sets.
	By means of example, consider a clique with three vertices $u$, $v$ and $w$
	can be represented as $(u \ltens v)\ltens w$ or $u\ltens(v\ltens w)$.
	
	We already observed the second reason in \Cref{rem:compVia},
	since graph symmetries allow us to represent the same graph by different decompositions,
	as shown in top-most modular decomposition below on the left.
	$$
	\begin{array}{cc}
		\gP\connn{u,v,w,t}=\vu1\quad\vv1\quad\vw1\quad\vt1\Gedges{u1/v1,v1/w1,w1/t1}=\gP\connn{t,w,v,u}
		\\
		\gP\connn{u,v,w,t}=\vu1\quad\vv1\quad\vw1\quad\vt1\Gedges{u1/v1,v1/w1,w1/t1}=\gP'\connn{u,w,v,t}
	\end{array}
	\quad\mbox{where}\quad
	\gP =\va1\quad\vb2\quad\vc3\quad\vd4\Gedges{a1/b2,b2/c3,c3/d4}
	\quand
	\gP'=\va1\quad\vc3\quad\vb2\quad\vd4\Gedges{a1/c3,b2/c3,b2/d4}
	\;.
	$$
	Finally, two symmetric prime graphs could provide distinct modular decompositions of the same graph,
	as shown above with symmetric prime graphs $\gP$ and $\gP'$.
\end{remark}

The first problem could be addressed by considering in the modular decomposition not only prime graphs, but also cliques and stable sets, that is, including $n$-ary versions of the operations $\lpar$ and $\ltens$.
We show later in this paper that this problem is irrelevant due to the associativity of $\lpar$ and $\ltens$.
The second problem cannot be addressed without enforcing a cumbersome order over graphs taking into account vertex labels and factor positions.
However, we can address the latter source of ambiguity by introducing the notion of \emph{base} of \emph{graphical connectives},
allowing us to provide a single canonical prime graph for each class of symmetric prime graphs.

\begin{definition}
	A \defn{graphical connective} 
	$\gC=\tuple{\vC,\uC}$
	(with \defn{arity} $n=\sizeof\vC$)
	is given by 
	a finite list of vertices $\vC=\tuple{v_1,\ldots,v_n}$ 
	and
	a non-reflexive symmetric edge relation $\uC$ over the set of vertices occurring in $\vC$.
	We denote by $\gG_\gC$ the graph corresponding to $\gC$, 
	that is, the graph $\gG_\gC=\tuple{\set{v\mid v \mbox{ in } \vC},\emptyset,\uC}$.
	The \defn{composition-via} a graphical connective 
	is defined as the composition-via the graph $\gG_\gC$.
	
	A graphical connective is \defn{prime} 
	if $\gG_\gC$ is a prime graph.	
	A set $\primesign$ of prime graphical connectives  is a \defn{base}
	if for each prime graph $\gP$ 
	there is a unique connective $\gC\in\primesign$ 
	such that $\gP\sim \gG_\gC$.
	
	Given an $n$-ary connective $\gC$,
	we define the following sets of permutations over the set $\intset1n$:
	\begin{equation}
		\begin{array}{l@{\;:\;}c@{\quad\coloneqq\quad}l}
			\mbox{the \defn{group\footnote{
						It can be easily shown that $\permset n$ 
						is subgroup of the group of permutations over the set $\intset1n$.
					}
				of symmetries} of $\gC$}&
			\Sym{\gC}	&
			\Set{\;\sigma\mid \gC\connn{a_1,\ldots, a_n}\greq\gC\connn{a_{\sigma(1)},\ldots, a_{\sigma(n)}}}
			\qquad
			\\
			\mbox{the \defn{set of dualizing symmetries} of $\gC$}&
			\Dsym{\gC}	&
			\Set{\;\sigma\mid \cneg{(\gC\connn{a_1,\ldots, a_n})}\greq\gC\connn{\cneg a_{\sigma(1)},\ldots, \cneg a_{\sigma(n)}}
			}
		\end{array}
	\end{equation}
	for any $a_1,\ldots, a_n$ single-vertex graphs.
\end{definition}

\begin{nota}\label{nota:prime}
	\def\myskip{\hskip2em}
	We define the following graphical connectives (with $n>1$):
	\begin{equation}\label{eq:gCon}
		\begin{array}{r@{\coloneqq}l@{\;=\;}c}
			\lpar\connn{v_1,v_2}
			&
			\Tuple{\tuple{v_1,v_2},\emptyset}
			&
			\vmod1{\vind1\quad \vind2}
			\\[5pt]
			\ltens\connn{v_1,v_2}
			&
			\Tuple{\tuple{v_1,v_2},\set{v_1v_2}}
			&
			\vmod2{\vind1 \quad \vind2 \Gedges{v1/v2}}
			\\[5pt]
			\gPath n\connn{v_1,\ldots,v_n}
			&
			\Tuple{\tuple{v_1,\ldots, v_n},\set{v_iv_{i+1}\mid i\in\intset1{n-1}}}
			&
			\vmod1{\vind1\quad\vind2\quad \vgdots1 \quad\vind{n-1}\quad\vind n
				\Gedges{v1/v2,v2/gdots1,gdots1/vn-1,vn-1/vn}}
			\\[5pt]
			\Pbull\connn{v_1,\ldots, v_5}
			&
			\Tuple{\tuple{v_1,\ldots, v_5},\set{(v_1v_2,v_2v_3,v_3v_4,v_5v_2,v_5v_3)}}
			&
			\vmod1{
				\begin{array}{c@{\quad}ccc@{\quad}c}
					\vind1&\vind2&&\vind3&\vind4
					\\
					&&\vind5
				\end{array}
				\Gedges{v1/v2,v2/v3,v3/v4}
				\Gedges{v5/v2/20,v3/v5/20}
			}
		\end{array}
	\end{equation}
\end{nota}

\begin{example}\label{ex:modular}
	Consider the following graph $\gG$ and its dual $\nG$:
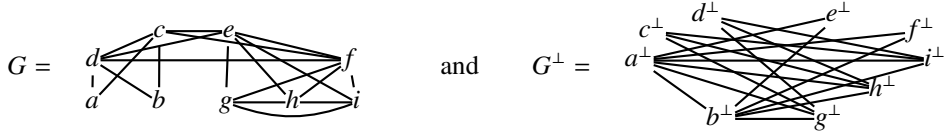
	$$
	\gG=
	\begin{array}{cc}
		\begin{array}{c@{\qquad}c}&\vc1\\\vd1\end{array}
		&
		\begin{array}{c@{\qquad\qquad}c}\ve1\\&\vf1\end{array}
		\\[10pt]
		\begin{array}{c@{\qquad}c}\va1&\vb1\end{array}
		&
		\begin{array}{c@{\qquad}c@{\qquad}c}\vg1& \vh1& \vi1\end{array}
	\end{array}
	\Gedges{e1/f1,c1/d1}
	\multiGedges{a1,b1}{c1,d1}
	\multiGedges{c1,d1}{e1,f1}
	\multiGedges{e1,f1}{g1,h1,i1}
	\Gedges{g1/h1,i1/h1}
	\bentGedges{i1/g1/20}
	\qquand
	\nG=
	\begin{array}{c@{\qquad}c}
		\begin{array}{c@{\quad}c}&\vnd1\\[-5pt]\vnc1\end{array}
		&
		\begin{array}{c@{\qquad}c}\vne1\\[-5pt]&\vnf1\end{array}
		\\
		\begin{array}{c@{\qquad}c}\vna1\\\\&\vnb1\end{array}
		&
		\begin{array}{c@{\quad}c@{\quad}c}&& \vni1 \\& \vnh1 \\ \vng1 \end{array}
	\end{array}
	\Gedges{na1/nb1}
	\multiGedges{ng1,nh1,ni1}{nc1,nd1}
	\multiGedges{na1,nb1}{ne1,nf1}
	\multiGedges{na1,nb1}{ng1,nh1,ni1}
	$$
	We can write them as 
	$$
	\begin{array}{l@{=}c@{=}c}
		\gG
		&
		\Pfour\Connn{a\lpar b, c\ltens d,e\ltens f,g\ltens(h\ltens i)}
		&
		\vmod1{\va1\quad\vb1}
		\quad
		\vmod2{\vc1\quad\vd1}
		\quad
		\vmod3{\ve1\quad\vf1}
		\quad
		\vmod4{\vg1\quad \vmod5{\vh1\quad \vi1}}
		\Gedges{e1/f1,c1/d1}
		\Gedges{mod1/mod2,mod2/mod3,mod3/mod4}
		\Gedges{g1/mod5,i1/h1}
		\\
		\nG
		&
		\nPfour\Connn{\cneg a\ltens \cneg b, \cneg c\lpar \cneg d,\cneg e\lpar \cneg f,\cneg g\lpar(\cneg h\lpar \cneg i)}
		\\&
		\Pfour\Connn{\cneg c\lpar \cneg d,\cneg a\ltens \cneg b, \cneg g\lpar(\cneg h\lpar \cneg i),\cneg e\lpar \cneg f}
		&
		\vmod3{\vne1\quad\vnf1}
		\quad
		\vmod1{\vna1\quad\vnb1}
		\quad
		\vmod4{\vng1\quad \vnh1\quad \vni1}
		\quad	
		\vmod2{\vnc1\quad\vnd1}
		\Gedges{na1/nb1}
		\Gedges{mod1/mod3,mod2/mod4,mod1/mod4}
	\end{array}
	$$
\end{example}

We can reformulate the standard result on modular decomposition as follows.
\begin{theorem}\label{thm:conDec}
	Let $\gG$ be a non-empty graph and $\primesign$ a base.
	Then then there is a unique way (up to symmetries of graphical connectives and associativity of $\lpar$ and $\ltens$)
	to write $\gG$ using single-vertex graphs and the graphical connectives in $\primesign$.
\end{theorem}
\begin{corollary}
	Two graphs are isomorphic iff they admit a same modular decomposition.
\end{corollary}



\subsection{Graphs as Formulas}

In order to represent graphs as formulas,
we define new connectives beyond conjunction and disjunction to represent graphical connectives in a base $\primesign$.
From now on, we assume bases $\primesign$ containing the graphical connectives in \Cref{eq:gCon} to be fixed.

\begin{definition}\label{def:formulas}
	The set of \defn{\forms} is generated by 
	the set of propositional atoms $\atomset$, 
	a \defn{unit} $\funit$, 
	using the following syntax:
	\begin{equation}
		\fA_1,\ldots,\fA_n
		\coloneqq 
		\funit	\mid 
		a 		\mid 
		\cneg a \mid 
		\con[\gP]\connn{\fA_1,\ldots, \fA_{n_\gP}}	
		\qquad
		\mbox{with $a\in\atomset$ and $\gP \in\primesign$}
	\end{equation}
	We simply denote $\lpar$ (resp. $\ltens$)
	the binary connective $\con[\lpar]$ (resp. $\con[\ltens]$)
	and 
	we write
	$\fA\lpar\fB$ instead of  $\con[\lpar]\connn{\fA,\fB}$
	(resp. $\fA\ltens\fB$ instead of $\con[\ltens]\connn{\fA,\fB}$).
	The \defn{arity} of the connective $\con[\gP]$ is the arity $n_\gP$ of $\gP$.
	
	A \defn{literal} is a \form of the form $a$ or $\cneg a$ for an atom $a\in\atomset$. 
	The set of literals is denoted $\litset$.
	A \defn{$\con$-formula} is a \form  with \defn{main connective} $\con$, that is, a \form of the form $\con\connn{\fA_1, \ldots, \fA_n}$.
	A \form is \defn{unit-free} if it contains no occurrences of $\funit$ and 
	\defn{vacuous} if it contains no atoms.
	A \form is \defn{pure} if non-vacuous and such that its vacuous subformulas are $\funit$.
	A \defn{\mllform} is a formula containing only occurrences of $\lpar$ and $\ltens$ connectives.

	A \defn{context formula} (or simply \defn{context}) $\fcC\conso$ is a formula containing an \defn{hole} $\chole$ taking the place of an atom. 
	Given a context $\fcC\conso$, the formula $\fcC\cons\fA$ is defined by simply replacing the atom $\chole$ with the formula $\fA$. 
	For example, if $\fcC\conso=\fB\lpar (\chole\ltens \fC)$, then $\fcC\cons\fA=\fB\lpar ( \fA\ltens \fC)$.
	
	For each $\fA$ formula (or context), the graph $\gof \fA$ is defined as follows:
	\begin{equation}\label{def:gof}
		\gof\chole = \chole
		\qquad
		\gof\funit = \emptygraph
		\qquad
		\gof a = a
		\qquad
		\gof {\cneg\fA} = \cneg{\gof\fA}
		\qquad
		\gof{\con[\gP]\connn{\fA_1,\ldots, \fA_n}}=
		\gP\Connn{\gof{\fA_1},\ldots, \gof{\fA_n}}
	\end{equation}
\end{definition}

\begin{nota}
	We could consider
	a formula $\fA$ over the set of occurrences of literals $\set{\fx_1,\ldots ,\fx_n}$ 
	as a \defn{synthetic connective}.
	That is, 
	we may denote by $\fA\connn{\fB_1,\ldots, \fB_n}$ 
	the formula obtained by replacing each literal $x_i$ with a corresponding $\fB_i$ for all $i\in\intset1n$.
	The set of \defn{symmetries} of $\fA$ (denoted $\Sym{\fA}$) is the set of permutations $\sigma$ over $\intset1n$ such that $\gof{\fA\connn{\fx_1,\ldots ,\fx_n}}=\gof{\fA\connn{\fx_{\sigma(1)},\ldots ,\fx_{\sigma(n)}}}$.
\end{nota}

\begin{definition}\label{def:feq}
	The equivalence relation $\feq$ over formulas
	is generated by the 
	following equations:
	\begin{equation*}
		\begin{array}{l@{\;}l}
			\mbox{\defn{Equivalence laws}}
			&
			\left\{\begin{array}{l@{\quad}r@{\;\feq\;}l@{\quad}l}
				&
				\con[\gP]\connn{\fA_1, \ldots,\fA_{\sizeof\gP}}
				&
				\con[\gP]\connn{\fA_{\sigma(1)},\ldots,\fA_{\sigma(\sizeof\vP)}}
				&
				\mbox{ for each }\sigma\in\Sym\gP
				\\
				&
				\fA\ltens(\fB\ltens \fC)
				&
				(\fA\ltens\fB)\ltens \fC
				\\
				&
				\fA\lpar(\fB\lpar \fC)
				&
				(\fA\lpar\fB)\lpar \fC
			\end{array}\right.
			\\\\
			\mbox{\defn{De-Morgan laws}}
			&
			\left\{\begin{array}{l@{\quad}r@{\;\feq\;}l@{\;}l}
				\multicolumn{4}{c}{
					\cneg\funit
					\;\feq\;
					\funit	
					\qquad\qquad
					\cnegneg \fA
					\;\feq\; 
					\fA
					\qquad
				}
				\\
				\mbox{only if }\Dsym\gP=\emptyset:
				&
				\cneg{\left(\con[\gP]\connn{\fA_1,\ldots\fA_{n_\gP}}\right)}
				&
				\con[\cneg{\gP}]\connn{\cneg{\fA_{\sigma(1)}},\dots,\cneg{\fA_{\sigma(n_\gP)}}}
				&
				\\
				\mbox{only if }\Dsym\gP\neq \emptyset:
				&
				\cneg{\left(\con[\gP]\connn{\fA_1,\ldots\fA_{n_\gP}}\right)}
				&
				\con[\gP]\connn{\cneg{\fA_{\rho(1)}},\dots,\cneg{\fA_{\rho(n_\gP)}}} 
				&
				\mbox{ for each }\rho\in\Dsym\gP
			\end{array}\right.
		\end{array}
	\end{equation*}
	for each $\gP\in\primesign$ (with arity $n_\gP$).
	
	The \defn{(linear) negation} over formulas is defined by letting 
	$$
	\cneg \funit = \funit
	\qquad
	\cnegneg \fA= \fA
	\qquad
	\cneg{\left(\con[\gP]\connn{\fA_1,\ldots, \fA_n}\right)}
	=
	\con[\gQ]\connn{\cneg\fA_{\sigma_\gP(1)}, \ldots, \cneg\fA_{\sigma_\gP(n)}}
	\quad
	$$
	where $\gQ$ is the unique graphical connective in $\primesign$ such that $\gof{\con[\gP]\Connn{a_1,\ldots, a_n}}=\gQ\Connn{\cneg a_{\sigma(1)},\ldots, \cneg a_{\sigma_n}}$
	for any single-vertex graphs $\cneg a_1,\ldots,\cneg a_n$ (with vertex labeled by $\cneg a_1,\ldots,\cneg a_n$ respectively)
	and
	a permutation $\sigma_\gP$ over the set $\intset1n$.
	\footnote{
		Note that the permutation $\sigma_\gP$ may be not unique.
		This is not a problem if we consider formulas up-to the equivalence relation $\feq$.
		Otherwise, in order to properly define the linear negation, 
		we should fix a permutation $\sigma_\gP$ for each graphical connective $\gP\in\primecon$ 
		in such a way 
		either $\sigma_P$ is an involution (in case $\gG_\gP\sim \cneg{\left(\gG_\gP\right)}$),
		or $\sigma_P\sigma_Q$ is the identity (in case $\gG_\gP\not\sim\cneg{\left(\gG_\gP\right)}\sim\gG_\gQ$ for a $\gQ\in\primecon\setminus\set\gP$).
	}
	
	The \defn{linear implication} $\fA\limp\fB$ is defined as $\cneg\fA\lpar\fB$,
	while
	the \defn{logical equivalence} $\fA\limpeq\fB$ is defined as $(\fA\limp\fB)\ltens(\fB\limp\fA)$.
\end{definition}

\begin{remark}
	As explained in \cite{acc:LMCS} (Section 9), the graphical connectives we discuss in this paper are \emph{multiplicative connectives} (in the sense of \cite{dan:reg:89,girard2000meaning,mai:19,acc:mai:20}) but they are not the same as the \emph{connectives-as-partitions} discussed in these works.
	In fact, 
	there is a unique $4$-ary graphical connective $\Pfour$ has symmetry group $\set{\idperm,(1,4)(2,3)}$,
	while, 
	as shown in 
	\cite{mai:19,acc:mai:20}, 
	there is a unique pair of dual ``primitive'' $4$-ary multiplicative connectives-as-partitions $\mathsf G_4$ and $\cneg{\mathsf G_4}$, 
	and 
	$\Sym{\Pfour}\subsetneq\Sym{\mathsf G_4}=\Sym{\cneg{\mathsf G_4}}$.
\end{remark}

%
%
	

The following result is consequence of \Cref{thm:modDec}.
\begin{proposition}\label{prop:feq}
	Let $\fA$ and $\fB$ be formulas. 
	If $\fA\feq\fB$, then $\gof\fA\greq\gof\fB$.
	Moreover, if $\fA$ and $\fB$ are unit-free, 
	then $\fA\feq\fB$ iff $\gof\fA\greq\gof\fB$.
\end{proposition}
Note that the the stronger result does not hold in presence of units. 
For an example consider the formulas $\funit\ltens\funit$ and $\funit\lpar\funit$.

\section{Sequent calculi over operating on graphs-as-formulas}\label{sec:calculi}

We assume the reader to be familiar with the definition of sequent calculus derivations as trees of sequents (see, e.g., \cite{troelstra:schwichtenberg:00}) but we recall here some definitions.

\begin{definition}
	We define a \defn{sequent} is a set of occurrences of formulas.
	A \defn{sequent system} $\Sys$ is a set of \defn{sequent rules} as the ones in \Cref{fig:rules}.
	In a sequent rule $\rho$, we say that a formula is \defn{active} if 
	it occurs in one of its premises 
	but not in its conclusion
	, and 
	\defn{principal} if it occurs in its conclusion but in none of its premises.%

	A \defn{proof} of a sequent $\Gamma$ is a derivation with no open premises, denoted $\vlderivation{\vlpr{\pi}{\Sys}{\Gamma}}$.
	We denote by 
	$\vlsmash{\vlderivation{\vlde{\pi'}{\Sys}{\Gamma}{\vlhy{\Gamma'}}}}$
	an \defn{(open) derivation} of $\Gamma$ from $\Gamma'$, 
	that is, 
	is a proof tree having exactly one open premise $\Gamma'$.
	
	A rule is \defn{admissible} in $\Sys$ 
	if there is a derivation of the conclusion of the rule 
	whenever all premises of the rule are derivable.
	A rule is \defn{derivable} in $\Sys$, 
	if there is a derivation in $\Sys$ from the premises to the conclusion of the rule.
\end{definition}

\begin{nota}
	In this paper we use the same notation to denote a sequent system $\Sys$ 
	and
	the set of formulas admitting a proof in $\Sys$.
\end{nota}

\begin{figure}[t]
	\centering
	\adjustbox{max width=\textwidth}{$
		\begin{array}{c}
			\vlinf{\axrule}{}{\vdash a , \cneg a}{}
			\qquad
			\vlinf{\lparr}{}{\vdash\Gamma, \fA\lpar \fB}{\vdash \Gamma ,\fA,\fB}
			\qquad
			\vliinf{\ltensr}{}{\vdash \Gamma, \fA\ltens \fB, \Delta}{\vdash \Gamma, \fA}{\vdash \fB,\Delta}
			\\\\
			\vliiinf{\dconr[\con]}{
				\text{\scriptsize$\begin{cases}\sigma\in\Sym{\con}\\\tau\in\Sym{\cneg\con}\end{cases}$}
			}{
				\vdash\Gamma_1, \ldots, \Gamma_n, \con\connn{\fA_1, \ldots, \fA_n}, \ncon\connn{\fB_{1}, \ldots \fB_{n}}
			}{
				\vdash\Gamma_1, \fA_{\sigma(1)}, \fB_{\tau(1)}
			}{
				\qquad\cdots\qquad
			}{
				\vdash\Gamma_n, \fA_{\sigma(n)},\fB_{\tau(n)}
			}		
			\\\\\hline\\
			%
			\qquad
			\vliinf{\mixr}{}{\vdash\Gamma_1, \Gamma_2}{\vdash\Gamma_1}{\vdash\Gamma_2}
			\qquad
			\vliinf{\wdtr}{
			}{
				\vdash \Gamma, \Delta,
				\con\connn{\fA_1,\ldots, \fA_n}
			}{
				\vdash \Gamma, \fA_k
			}{
				\vdash \Delta,\con\connn{\fA_1,\ldots, \fA_{k-1},\funit,\fA_{k+1},\ldots,\fA_n}
			}
			\\\\
			\vlinf{\unitor}{
				\dagger
			}{
				\vdash \Gamma, 
					\con\connn{\fA_1,\ldots, \fA_k,\funit,\fA_{k+1},\ldots,\fA_n}
			}{
				\vdash \Gamma,\fC\connn{\fA_{\sigma(1)},\ldots, \fA_{\sigma(n)}}
			}
			\\\\
			\dagger
			\coloneqq
			\sigma\in\Sym{\fC}
			\quand	
			\gof{\con\connn{\fA_1,\ldots, \fA_k,\funit,\fA_{k+1},\ldots,\fA_n}}
			=
			\gof{\fC\connn{\fA_{\sigma(1)},\ldots, \fA_{\sigma(n)}}}
			\neq
			\unit
			%
		\end{array}
		$}
	\caption{Linear sequent calculus rules for $\PML$ and $\PMLx$.}
	\label{fig:rules}
\end{figure}

%
%

\begin{definition}\label{def:GML}
	We define the following sequent systems using the rules in \Cref{fig:rules}.
	\begin{equation}\label{eq:linSyss}
		\begin{array}{ll@{\;=\;}l}
			\mbox{\defn{Multiplicative Graphical Logic}}\colon\;
			&
			\PML
			&
			\Set{\axrule,\lparr,\ltens,\dconr[\gP] \mid  \gP\in\primecon}
			\\
			\mbox{\defn{Multiplicative Graphical Logic with mix}}\colon\;
			&
			\PMLx
			&
			\PML\cup\Set{\mixr,\wdtr,\unitor}
			\\
		\end{array}
	\end{equation}
\end{definition}
\begin{obs}[Rules Exegesis]
	The rules \defn{axiom} ($\axrule$), \defn{par} ($\lparr$), \defn{tensor} ($\ltensr$), \defn{cut} ($\cutr$), and \defn{mix} ($\mixr$) 
	are the standard as in multiplicative linear logic with mix.
	Note that $\axrule$ is restricted to atomic formulas.

	The \defn{dual connectives} rule ($\dconr$) handles a pair of dual connectives at the same time.%
	\footnote{
		Rules handling multiple operators at the same time are not a novelty in structural proof theory:
		in focused proof systems (see, e.g.\cite{andreoli:92,miller:07:cslb,mil:pim:13})
		rules can handle multiple connectives of a same formula,
		while
		in modal logic and linear logic (see, e.g.,  \cite{girard:98,blackburn:venema:modal,bru:str:tableaux09,lel:pim:seq})
		is quite standard to have rules handling modalities occurring in different formulas of a same sequent.
	}
	To get an intuition of this rule, 
	consider the right-conjunction rule ($\land_R$) used in two-sided sequent calculi for classical logic shown below on the left.
	The interpretation of this rule is that if the left premise \emph{and} the right premise are derivable, then the conclusion is.
	Note that, even if the rule does not introduce a conjunction on the lefthand-side of the $\vdash$, 
	the interpretation of the conclusion sequent is the same of the interpretation of the sequent in which $\fA_1$ and $\fA_2$ are in conjunction 
	because the standard interpretation of a two-sides sequent $\Gamma \vdash \Delta$ is defined as 
	$\left(\bigwedge_{\fA\in\Gamma}\cneg\fA \right)\lor \left(\bigvee_{\fB\in\Delta}\fB\right)$.
	\begin{equation*}
		\adjustbox{max width=\textwidth}{$
			\begin{array}{c|c}
				\vliiinf{\land_R}{}{
					\Gamma_1,\Gamma_2, \fA_1, \fA_2
					\vdash
					\fB_1\land\fB_2 , \Delta_1,\Delta_2
				}{
					\vmod1{\Gamma_1, \fA_1\vdash\fB_1,\Delta_1}
				}{\mbox{``and''}}{
					\vmod2{\Gamma_2, \fA_2\vdash\fB_2,\Delta_2}
				}
				&
				\vlinf{}{}{
					\Gamma_1,\Gamma_2,\Gamma_3, \Gamma_4, \con[\Pfour]\connn{\fA_1,\fA_2,\fA_3,\fA_4}
					\vdash
					\con[\Pfour]\connn{\fB_1,\fB_2,\fB_3,\fB_4}, 
					\Delta_1,\Delta_2,\Delta_3,\Delta_4
				}{
					\Pfour\Connn{
					\vmod1{ \Gamma_1, \fA_1\vdash\fB_1,\Delta_1}
					\;,\;
					\vmod2{ \Gamma_2, \fA_2\vdash\fB_2,\Delta_2}
					\;,\;
					\vmod3{ \Gamma_3, \fA_3\vdash\fB_3,\Delta_3}
					\;,\;
					\vmod4{ \Gamma_4, \fA_4\vdash\fB_4,\Delta_4}
					}
				}
			\end{array}
		$}
	\end{equation*}
	In a two-sided setting the rule $\dconr$ could have been reformulated by introducing the same connective in both sides.
	Intuitively, such a rule would internalize in the logic a meta-connective establishing a relation between the premises of the rule,
	as intuitively shown above on the right for the connective $\Pfour$.
	
	The names of the rules 
	\defn{unitor} ($\unitor$) 
	and
	\defn{weak-distributivity} ($\wdtr$)  
	are inspired by 
	the literature of \emph{monoidal categories} \cite{maclane:71}
	and
	\emph{weakly distributive categories} \cite{seely:89,cockett:seely:97,cockett:seely:97:mix}.
	The rule $\unitor$ internalize the fact that the unit $\funit$ is the neutral element for all connectives 
	(its side condition prevents the creation of non-pure formulas).
	Under the assumption of the existence of a $\funit$ which is the unit of both $\ltens$ and $\lpar$,
	the rule $\wdtr$ generalizes the \emph{weak-distribution law} (shown below on the left) of the $\ltens$ over the $\lpar$ to the weak-distributivity of $\ltens$ over any connective (see below on the top-right)
	\begin{equation}\label{eq:WD}
		\adjustbox{max width=.92\textwidth}{$
			\begin{array}{c|c}
				\fA\ltens(\fB\lpar\fC)\longrightarrow(\fA\ltens\fB)\lpar \fC
				&
				\begin{array}{c@{\;\longrightarrow\;}c}
					\fC\ltens\con\connn{\fA_1,\ldots, \fA_{k}, \fB,\fA_{k+1},\ldots, \fA_n}
					&
					\con\connn{\fA_1,\ldots, \fA_k, \fB\ltens \fC,\fA_{k+1},\ldots, \fA_n}
					\\
					\con\connn{\fA_1,\ldots, \fA_k, \fB\lpar\fC,\fA_{k+1},\ldots, \fA_n}
					&
					\con\connn{\fA_1,\ldots, \fA_k, \fB,\fA_{k+1},\ldots, \fA_n}\lpar \fC
				\end{array}
			\end{array}
		$}
	\end{equation}
	Note that an additional law is required to formalize the weak-distributivity law of all connectives over $\lpar$ (see above on the bottom-right). This law corresponds to the rule $\wdpr$ in \Cref{fig:admRules}.
\end{obs}
\begin{nota}\label{nota:noEX}
	Unless strictly needed for sake of clarity, 
	we omit to the permutations over the indices of the subformulas in rules.
\end{nota}
\subsection{Properties of the systems $\PML$ and $\PMLx$}\label{sec:propisomix}

\begin{figure}[t]
	\centering
	\adjustbox{max width=\textwidth}{$
		\begin{array}{c}
			\vlinf{\AXrule}{\text{\scriptsize $\fA$ pure}}{\vdash \fA , \nfA}{}
			\qquad
			\vliinf{\cutr}{}{\vdash\Gamma_1, \Gamma_2}{\vdash\Gamma_1,\fA}{\vdash\Gamma_2, \cneg \fA}
			\qquad
			\vlinf{\wdpr}{}{\vdash\Gamma, \con\connn{\funit,\fB_1,\ldots, \fB_n}, \fA}{\vdash\Gamma, \con\connn{\fA,\fB_1,\ldots, \fB_n}}
			\\\\
			\vliinf{\deepr}{\text{\scriptsize$\gof{\fcC\consfu}=\gof\fB$}}{
				\vdash \Gamma, \Delta,
				\fcC\cons{\fA}
			}{
				\vdash \Gamma, \fA
			}{
				\vdash \Delta,\fB
			}
			\qquad
			\vliiinf{\dconr[\fC]}{
				\text{\scriptsize $\begin{cases}\sigma\in\Sym{\fC}\\\tau\in\Sym{\nfC}\end{cases}$}
			}{
				\vdash\Gamma_1, \ldots, \Gamma_n, \fC\connn{\fA_1, \ldots, \fA_n}, \nfC\connn{\fB_{1}, \ldots \fB_{n}}
			}{
				\vdash\Gamma_1, \fA_{\sigma(1)}, \fB_{\tau(1)}
			}{
				\qquad\cdots\qquad
			}{
				\vdash\Gamma_n, \fA_{\sigma(n)},\fB_{\tau(n)}
			}
		\end{array}
		$}
	\caption{Admissible rules in $\PMLx$.
	}
	\label{fig:admRules}
\end{figure}

We start by observing that these systems are \emph{initial coherent}~\cite{avr:canonical:01,mil:pim:13}, that is, we can derive the implication $\fA\limp\fA$ for any formula $\fA$ only using atomic axioms.
%
To prove this result we observe that the generalized version of $\dconr$ (that is, the rule $\dconr[\fC]$) is derivable
by induction on the structure of $\fC$ using the rule $\dconr$.
Therefore, we can prove that the generalized non-atomic axiom rule ($\AXrule$) is derivable, and that both $\PML$ and $\PMLx$ are initial coherent
\begin{lemma}\label{lem:gdcon}
	Let $\fC$ be a pure formula.
	Then rule $\dconr[\fC]$ is derivable.
\end{lemma}
\begin{proof}
	By induction on the structure of $\fC$:
	\begin{itemize}
		\item 
		if $\fA=a$ is a literal, then $\AXrule$ is an instance of $\axrule$;
		\item 
		if $\fA=\con\connn{\fB_1,\ldots, \fB_k,\funit,\fB_{k+1},\ldots, \fB_n}$, then apply twice $\unitor$ to the sequent $\vdash \fA ,\nfA$ to obtain the sequent of pure formulas  $\vdash \con[\fC]\connn{\fB_1,\ldots, \fB_n},\con[\nfC]\connn{\nfB_1,\ldots, \nfB_n}$.
		We conclude by inductive hypothesis;
		\item 
		if $\fA=\con\connn{\fB_1,\ldots, \fB_n}$ and $\fB_i\neq\funit$ for all $i\in\intset1n$,
		then apply the rule $\dconr$ to obtain sequents of pure formulas the form $\fB_i,\nfB_i$ for all $i\in\intset1{\aryof{\con}}$. 
		We conclude by inductive hypothesis.
	\end{itemize}
\end{proof}
\begin{corollary}\label{lem:gAX}
	The rule $\AXrule$ is derivable in $\PML$ and in $\PMLx$.
\end{corollary}
\begin{theorem}\label{thm:gAX}
	The systems $\PML$ and $\PMLx$ are initial coherent (with respect to pure formulas).
\end{theorem}

We then prove the admissibility of $\cutr$ via \emph{cut-elimination} by providing a cut-elimination procedure.

\begin{theorem}[Cut-elimination]\label{thm:cutelim}
	Let $\XS\in\set{\PML,\PMLx}$.
	The rule $\cutr$ is admissible in $\XS$.
\end{theorem}
\begin{proof}
	We define the \emph{size} of a formula as sum of the number of $\funit$, connectives and twice the number of literals in it.
	The \emph{size} of a derivation is the sum of the sizes of the active formulas in all $\cutr$-rules.
	The result follows by the fact that each \emph{cut-elimination step} from \Cref{fig:cut-elimPML,fig:cut-elimPMLx} reduces the size of a derivation.

	Note that in order to ensure that both active formulas of a $\cutr$ are principal with respect to the rule immediately above it we also need to consider the \emph{commutative} cut-elimination steps from \Cref{fig:cut-elimCom}.
	The treatment of these rule, as well as the definition of a size taking into account them, is not covered in the detail here because it is standard in the literature (see, e.g., \cite{troelstra:schwichtenberg:00}).
\end{proof}
\begin{corollary}\label{cor:impTrans}
	Let $\XS\in\set{\PML,\PMLx}$.
	If $\proves[\XS]\fA\limp \fB$ and $\proves[\XS]\fB\limp \fC$, then $\proves[\XS]\fA\limp \fC$.
\end{corollary}


\def\downsquigarrow{\rotatebox{-90}{$\rightsquigarrow$}}
\begin{figure*}
	\centering
	\adjustbox{max width=\textwidth}{$\begin{array}{c}
			\vlderivation{
				\vliin{\cutr}{}{\vdash a,\Gamma}{
					\vlin{\axrule}{}{\vdash a,\cneg a}{\vlhy{}}
				}{
					\vlhy{\vdash a, \Gamma}
				}
			}
			\;\rightsquigarrow\; 
			\vdash \cneg a,\Gamma
			\qquad
			%
			\vlderivation{
				\vliin{\cutr}{}{
					\vdash\Gamma,\Delta,\Sigma
				}{
					\vliin{\ltensr}{}{
						\vdash\Gamma, \Delta, \fA\ltens\fB 
					}{
						\vlhy{\vdash\Gamma,\fA}
					}{
						\vlhy{\vdash\Delta,\fB}
					}
				}{
					\vlin{\lparr}{}{
						\vdash \Sigma, \nfA \lpar \nfB
					}{
						\vlhy{\vdash\Sigma, \nfA\lpar\nfB}
					} 
				}
			}
			\;\rightsquigarrow\; 
			\vlderivation{
				\vliin{\cutr}{}{
					\vdash\Gamma,\Delta,\Sigma
				}{
					\vlhy{\vdash\Gamma,\fA}
				}{
					\vliin{\cutr}{}{
						\vdash \Delta,\Sigma, \nfA
					}{
						\vlhy{\vdash\Delta,\fB}	
					}{
						\vlhy{\vdash\Sigma, \nfA,\nfB}
					}
				}
			}		
			%
			\\\\\hline\\
			%
			\vlderivation{
				\vliin{\cutr}{}{
					\vdash\Gamma_1,\ldots, \Gamma_n, \Delta_1, \ldots, \Delta_n,\con\connn{\fA_1, \ldots, \fA_n},\ncon\connn{\fC_1,\ldots, \fC_n}
				}{
					\vliiin{\dconr}{}{
						\vdash\Gamma_1,\ldots, \Gamma_n, \con\connn{\fA_1, \ldots, \fA_n},\ncon\connn{\fB_1, \ldots, \fB_n}
					}{\vlhy{\vdash\Gamma_1,\fA_1,\fB_1}}{\vlhy{\qquad\cdots\qquad}}{\vlhy{\vdash\Gamma_n,\fA_n,\fB_n}}
				}{
					\vliiin{\dconr}{}{
						\vdash\Delta_1,\ldots, \Delta_n, \con\connn{\nfB_1, \ldots, \nfB_n},\ncon\connn{\fC_1,\ldots, \fC_n}
					}{\vlhy{\vdash\Delta_1,\nfB_1,\fC_1}}{\vlhy{\qquad\cdots\qquad}}{\vlhy{\vdash\Delta_n,\nfB_n,\fC_n}}
				}
			}
			%
			\\\downsquigarrow\\
			\vlderivation{
				\vliiin{\dconr}{}{
					\vdash\Gamma_1,\ldots, \Gamma_n, \Delta_1, \ldots, \Delta_n,\ncon\connn{\fB_1, \ldots, \fB_n},\con\connn{\fC_1,\ldots, \fC_n}
				}{
					\vliin{\cutr}{}{\vdash\Gamma_1,\Delta_1,\fA_1,\fC_1}{\vlhy{\vdash\Gamma_1,\fA_1,\fB_1}}{\vlhy{\vdash\Delta_1,\nfB_1,\fC_1}}
				}{
					\vlhy{\quad\cdots\quad}
				}{
					\vliin{\cutr}{}{\vdash\Gamma_n,\fA_n,\fC_n}{\vlhy{\vdash\Gamma_n,\fA_n,\fB_n}}{\vlhy{\vdash\Delta_n,\nfB_n,\fC_n}}
				}
			}
		\end{array}$}
	\caption{Cut-elimination steps for $\PML$.}
	\label{fig:cut-elimPML}
\end{figure*}

\begin{figure*}[t]
	\centering
	\adjustbox{max width=\textwidth}{$\begin{array}{c}
		\vlderivation{
			\vliin{\cutr}{}{
				\vdash\Gamma,\Delta
			}{
				\vlin{\unitor}{}{
					\vdash\Gamma, \fP\connn{\funit,\fA_2, \ldots, \fA_n}
				}{
					\vlhy{\vdash\Gamma, \fC\connn{\fA_2,\ldots, \fA_n}}
				}
			}{
				\vlin{\unitor}{}{
					\vdash \Delta, \nfP\connn{\funit,\nfA_2, \ldots, \nfA_n}
				}{
					\vlhy{\vdash\Delta, \nfC\connn{\nfA_2,\ldots, \nfA_n}}
				}
			}
		}
		\;\rightsquigarrow\; 
		\vlderivation{
			\vliin{\cutr}{}{
				\vdash\Gamma,\Delta
			}{
				\vlhy{\vdash\Gamma, \fC\connn{\fA_2,\ldots, \fA_n}}
			}{
				\vlhy{\vdash\Delta, \nfC\connn{\nfA_2,\ldots, \nfA_n}}
			}
		}
		\\\\\hline\\
		\vlderivation{
			\vliin{\cutr}{}{
				\vdash\Gamma_1,\Gamma_2,\Delta_1,\Delta_2
			}{
				\vliin{\wdtr}{}{
					\vdash\Gamma_1,\Gamma_2, \fP\connn{\fA_1, \ldots, \fA_n}
				}{
					\vlhy{\vdash\Gamma_1,\fA_1}
				}{
					\vlhy{\vdash\Gamma_2, \fP\connn{\funit,\fA_2,\ldots, \fA_n}}
				}
			}{
				\vliin{\wdtr}{}{
					\vdash\Delta, \nfP\connn{\nfA_1, \ldots, \nfA_n}
				}{
					\vlhy{\vdash\Delta_1,\nfA_1}
				}{
					\vlhy{\vdash\Delta_2, \nfP\connn{\funit,\nfA_2,\ldots, \nfA_n}}
				}
			}
		}
		\\\downsquigarrow\\\\
		\vlderivation{
			\vliin{\mixr}{}{
				\vdash\Gamma_1,\Gamma_2,\Delta_1,\Delta_2
			}{
				\vliin{\cutr}{}{
					\vdash\Gamma_1,\Delta_1
				}{
					\vlhy{\vdash\Gamma_1,\fA_1}
				}{
					\vlhy{\vdash\Delta_1,\nfA_1}
				}
			}{
				\vliin{\cutr}{}{
					\vdash\Gamma_2,\Delta_2
				}{
					\vlhy{\vdash, \Gamma_2,\fP\connn{\funit,\fA_2,\dots \fA_n}} 
				}{
					\vlhy{\vdash\Delta_2, \nfP\connn{\funit,\nfA_2,\dots \nfA_n}} 
				}
			}
		}
		%
		%
		\\\\\hline\\
		%
		\vlderivation{
			\vliin{\cutr}{}{
				\vdash\Gamma_1,\ldots, \Gamma_n, \Delta,\Sigma,\fP\connn{\fA_1, \ldots, \fA_n}
			}{
				\vliiin{\dconr}{}{
					\vdash\Gamma_1,\ldots, \Gamma_n, \fP\connn{\fA_1, \ldots, \fA_n},\nfP\connn{\fB_1, \ldots, \fB_n}
				}{
					\vlhy{\vdash\Gamma_1,\fA_1,\fB_1}
				}{
					\vlhy{\qquad\cdots\qquad}
				}{
					\vlhy{\vdash\Gamma_n,\fA_n,\fB_n}
				}
			}{
				\vliin{\wdtr}{}{
					\vdash \Delta,\Sigma, \fP\connn{\nfB_1, \ldots, \nfB_n}
				}{
					\vlhy{\vdash \Delta,\nfB_1}
				}{
					\vlhy{\vdash \Sigma, \fP\connn{\funit,\nfB_2,\ldots, \nfB_n}}
				}
			}
		}
		\\\downsquigarrow\\\\
		\vlderivation{
			\vliin{\wdtr}{}{
				\vdash\Gamma_1,\ldots, \Gamma_n, \Delta,\Sigma,\fP\connn{\fA_1, \ldots, \fA_n}
			}{
				\vliin{\cutr}{}{
					\vdash \Gamma_1, \Delta, \fA_1
				}{
					\vlhy{\vdash \Gamma_1, \fA_1,\fB_1}
				}{
					\vlhy{\vdash \Delta, \nfB_1}
				}
			}{
				\vliin{\cutr}{}{
					\vdash \Gamma_2,\ldots, \Gamma_n,\Sigma, \fP\connn{\funit,\fA_2, \ldots, \fA_n}
				}{
					\vlin{2\times \unitor}{}{
						\vdash\Gamma_2,\ldots, \Gamma_n, 
						\fP\connn{\funit,\fA_1, \ldots, \fA_n},
						\nfP\connn{\funit,\fB_1, \ldots, \fB_n}
					}{
						\vliin{\dconr[\fC]}{}{
							\vdash\Gamma_2,\ldots, \Gamma_n, 
							\con[\fC]\connn{\fA_1, \ldots, \fA_n},
							\ncon[\fC]\connn{\fB_1, \ldots, \fB_n}
						}{
							\vlhy{\vdash\Gamma_2,\fA_2,\fB_2\qquad\cdots\qquad}
						}{
							\vlhy{\vdash\Gamma_n,\fA_n,\fB_n}
						}
					}
				}{
					\vlhy{\vdash \Sigma, \fP\connn{\funit,\nfB_2,\ldots, \nfB_n}}
				}
			}
		}
	\end{array}$}
	\caption{
		Additional cut-elimination steps in $\GMLx$.
	}
	\label{fig:cut-elimPMLx}
\end{figure*}

\begin{figure*}[t]
	\centering
	\adjustbox{max width=\textwidth}{$\begin{array}{c}
		\vlderivation{
			\vliin{\cutr}{}{\vdash\Gamma_1, \Gamma_2,\Delta}{
				\vlin{\rho}{}{\vdash\Gamma_1,\Delta, \fA}{
					\vlhy{\vdash\Gamma_1,\Delta',\fA}
				}
			}{
				\vlhy{\vdash\nfA,\Gamma_2}
			}
		}
		\quad\rightsquigarrow\quad
		\vlderivation{
			\vlin{\rho}{}{\vdash\Gamma_1, \Gamma_2,\Delta'}{
				\vliin{\cutr}{}{\vdash\Gamma_1,\Gamma_2,\Delta}{
					\vlhy{\vdash\Gamma_1,\Delta',\fA}
				}{
					\vlhy{\vdash\nfA,\Gamma_2}
				}
			}
		}
		\\\\
		\vlderivation{
			\vliin{\cutr}{}{
				\vdash\Gamma_1,\ldots, \Gamma_{n+1},\Delta
			}{
				\vliiin{\rho}{}{
					\vdash\Gamma_1,\ldots, \Gamma_n, \Delta,\fA
				}{
					\vlhy{\vdash\Gamma_1, \Delta_1'}
				}{\vlhy{\cdots}}{
					\vlhy{\vdash\Gamma_n,\Delta_n',\fA}
				}
			}{
				\vlhy{\vdash\nfA,\Gamma_{n+1}}
			}
		}
		\quad\rightsquigarrow\quad
		\vlderivation{
			\vliiiin{\rho}{}{
				\vdash\Gamma_1,\ldots,\Gamma_{n+1},\Delta
			}{
				\vlhy{\vdash\Gamma_1, \Delta_1'}
			}{\vlhy{\cdots}}{
				\vlhy{\vdash\Gamma_{n-1},\Delta_{n-1}'}
			}{
				\vliin{\cutr}{}{
					\vdash\Gamma_n,\Gamma_{n+1},\Delta_n'
				}{
					\vlhy{\vdash\Gamma_n,\Delta_n',\fA}
				}{
					\vlhy{\vdash\Gamma_{n+1},\nfA}
				}
			}
		}
		\\\\
		\vlderivation{
			\vlin{\unitor}{}{
				\vdash
				\Gamma
				\fP\connn{\fA_1,\ldots, \fA_{i-1},\funit,\fA_{i+1},\ldots,\fA_{j-1},\funit,\fA_{j+1},\ldots,  \fA_n}
			}{
				\vlin{\unitor}{}{
					\vdash
					\Gamma, 
					\fcon{\gP}\connn{\fA_1,\ldots, \fA_{i-1},\funit,\fA_{i+1},\ldots,\fA_{j-1},\fA_{j+1},\ldots,  \fA_n}
				}{
					\vlhy{
						\vdash
						\Gamma,
						\fC\connn{\fA_1,\ldots, \fA_{i-1},\fA_{i+1},\ldots,\fA_{j-1},\fA_{j+1},\ldots,  \fA_n}}
				}
			}
		}
		\;\rightsquigarrow\;
		\vlderivation{
			\vlin{\unitor}{}{
				\vdash
				\Gamma,
				\fP\connn{\fA_1,\ldots, \fA_{i-1},\funit,\fA_{i+1},\ldots,\fA_{j-1},\funit,\fA_{j+1},\ldots,  \fA_n}
			}{
				\vlin{\unitor}{}{
					\vdash
					\Gamma,
					\fcon{\gP'}\connn{\fA_1,\ldots, \fA_{i-1},\fA_{i+1},\ldots,\fA_{j-1},\funit,\fA_{j+1},\ldots,  \fA_n}
				}{
					\vlhy{
						\vdash
						\Gamma,
						\fC\connn{\fA_1,\ldots, \fA_{i-1},\fA_{i+1},\ldots,\fA_{j-1},\fA_{j+1},\ldots,  \fA_n}}
				}
			}
		}
	\end{array}$}
	\caption{
		Commutative cut-elimination steps.
	}
	\label{fig:cut-elimCom}
\end{figure*}

The admissibility of the $\cutr$-rule implies analyticity of $\PML$ via the standard \emph{sub-formula property}, that is, all (occurrences of) formulas occurring in the premises of a rule are subformulas of the ones in the conclusion.
\begin{corollary}[Analyticity of $\PML$]
	Let $\Gamma$ be a sequent.
	If $\proves[\PML]\Gamma$, then there is a proof of $\Gamma$ in $\PML$ only containing occurrences of sub-formulas of formulas $\Gamma$.
\end{corollary}
However, the same result does not hold for $\PMLx$ because of the rule $\unitor$.
In fact, the presence of more-than-binary connectives and their units (in this case, a unique unit $\funit$) implies,
as observed in the previous works on graphical logic~\cite{acc:hor:str:LICS2020,acc:LMCS,acc:FSCD22},
the possibility of having \emph{sub-connectives}, that is, connectives with smaller arity behaving as if certain entries of the connective are fixed to be units.
\def\qsubf{quasi-subformula\xspace}
\begin{definition}
	Let $\gP$ and $\gQ$ be prime graphs.
	If 
	$\gP\connn{\funits, v_{i_1},\funits, v_{i_k},\funits}\isym\gQ\connn{v_1,\dots, v_n}$ 
	for single-vertex graphs $v_1, \ldots,v_n$
	and 
	for some distinct $i_1,\ldots,i_k\in\intset1n$,
	then we may write
	$\con[\gP\restr{i_1,\ldots, i_k}]=\con[\gQ]$
	and we say that
	the connective $\con[\gQ]$ is a \defn{sub-connective} of if $\con[\gP]$.

	A \defn{\qsubf} of a formula 
	$\fA=\con[\gP]\connn{\fB_1,\ldots,\fB_n}$ 
	is a formula of the form
	$\con[\gP'\restr{i_1,\ldots, i_k}]\connn{\fB'_{i_1},\ldots,\fB'_{i_k}}$
	with
	$\fB'_{i_j}$ a \qsubf of $\fB_{i_j}$ for all $i_j\in\intset{i_1}{i_k}$.
\end{definition}
\begin{corollary}[Analyticity of $\PMLx$]
	Let $\Gamma$ be a sequent.
	If $\proves[\PMLx]\Gamma$ then there is a proof of $\Gamma$ in $\PMLx$ only containing occurrences of \qsubf of formulas in $\Gamma$.
\end{corollary}

\begin{corollary}[Conservativity]
	The logic $\PML$ is a conservative extension of $\MLL$.
	The logic $\PMLx$ is a conservative extension of $\MLLx$.
\end{corollary}
\begin{proof}
	For $\PML$ it is consequence of the subformula property.
	For $\PMLx$ it suffices to remark that $\lpar$ and $\ltens$ have no sub-connectives, therefore \qsubf are simply sub-formulas.
\end{proof}

%
%

For both $\PML$ and $\PMLx$ we have the following result which takes the name of \emph{splitting} in the deep inference literature (see, e.g, \cite{gug:tub:split,gug:str:01,gug:str:02}).
This result states that is always possible, during proof search, to apply a rule removing a connective after having applied certain rules in the context.%
\footnote{
	Note that in the linear logic literature, the splitting lemma is usually formulated as a special case of the lemma here, where a $\ltens$ is removed without requiring the application of rules to the context.
}
%

\begin{lemma}[Splitting]\label{lem:seqSplit}
	Let 
	$\Gamma, \con\connn{\fA_1, \ldots, \fA_n}$ be a sequent
	and 
	let
	$\XS\in\set{\PML,\PMLx}$.
	If $\proves[\XS]\Gamma,\con\connn{\fA_1, \ldots, \fA_n} $,
	then there is
	a derivation of the following shape
	\begin{equation*}
		\vlderivation{
			\vlde{\pi_0}{}{
				\vdash \Gamma,\con\connn{\fA_1, \ldots,\fA_{k-1},\funit,\fA_{k+1}, \fA_n}
			}{
				\vlin{\unitor}{}{
					\vdash \Gamma',\con\connn{\fA_1, \ldots,\fA_{k-1},\funit,\fA_{k+1}, \fA_n}
				}{
					\vlpr{\pi_1}{}{\vdash \Gamma',\fC\connn{\fA_1, \ldots,\fA_{k-1},\fA_{k+1}, \fA_n}}
				}
			}
		}	
		\qquor
		\vlderivation{
			\vlde{\pi_0}{}{
				\vdash \Gamma,\con\connn{\fA_1, \ldots, \fA_n}
			}{\vliiin{\rho}{}{
					\vdash \Gamma',\con\connn{\fA_1, \ldots, \fA_n}
				}{
					\vlpr{\pi_1}{}{\vdash \Delta_1,\fA_1}
				}{
					\vlhy{\cdots}
				}{
					\vlpr{\pi_n}{}{\vdash \Delta_n,\fA_n}
				}
			}
		}
		\mbox{ with  $\rho \in\set{\lparr,\ltensr,\dconr}$}
	\end{equation*}
\end{lemma}
\begin{proof}
	By case analysis of the last rule occurring in a proof $\pi$ of $\Gamma, \con\connn{\fA_1, \ldots, \fA_n}$:
	\begin{itemize}
		\item 
		the last rule cannot be a $\axrule$ since $\con\connn{\fA_1, \ldots, \fA_n}$ contains at least one connective;
		
		\item if the last rule is a $\lparr$ or a $\unitor$, then 
		either this is the desired rule, or we conclude by inductive hypothesis on its premise;
		
		\item if the last rule is a $\mixr$, then we conclude by inductive hypothesis on the premise containing the formula $\con\connn{\fA_1, \ldots, \fA_n}$; 
		
		\item if the last rule is in $\set{\ltensr,\dconr,\wdtr,\unitor}$
		then either this is the desired rule or one of the (provable) premises of this rule is of the shape $\Gamma',\con\connn{\fA_1, \ldots, \fA_n}$,
		allowing us to conclude by inductive hypothesis.
		\qedhere
	\end{itemize}
\end{proof}

We conclude this section proving the admissibility of the rule
$\wdpr$ which we use to simplify proofs in the next section.

\begin{lemma}\label{lem:wdpr}
	The rule $\wdpr$ is admissible in $\PMLx$.
\end{lemma}
\begin{proof}
	In \Cref{fig:wdprElim} we providde a procedure to remove (top-down) all occurrences of $\wdpr$.
	Similar to cut-elimination, we use the commutative steps from \Cref{fig:cut-elimCom}
	to ensure that the active formula of the $\wdpr$ we want to remove is principal with respect to the rule immediately above it.
\end{proof}

\begin{figure*}[t]
	\def\wdprred{\quad\rightsquigarrow\quad}
	\adjustbox{max width = \textwidth}{$\begin{array}{c}
		\vlderivation{
			\vlin{\wdpr}{}{
				\vdash \Gamma, \fA,\funit\lpar \fB
			}{
				\vlin{\lparr}{}{
					\vdash \Gamma, \fA\lpar \fB
				}{
					\vlpr{\pi_1}{}{\vdash \Gamma, \fA, \fB}
				}
			}
		}
		\wdprred
		\vlderivation{
			\vlin{\unitor}{}{
				\vdash \Gamma, \fA,\funit\lpar \fB
			}{
				\vlpr{\pi_1}{}{\vdash \Gamma, \fA,\fB}
			}
		}
		\qquad\qquad
		\vlderivation{
			\vlin{\wdtr}{}{
				\vdash \Gamma, \Delta, \fA,\funit\ltens \fB
			}{
				\vliin{\ltensr}{}{
					\vdash \Gamma, \Delta, \fA\ltens \fB
				}{
					\vlpr{\pi_1}{}{\vdash \Gamma, \fA}
				}{
					\vlpr{\pi_2}{}{\vdash \Delta, \fB}
				}
			}
		}
		\wdprred
		\vlderivation{
			\vlin{\unitor}{}{
				\vdash \Gamma, \Delta, \fA, \funit\ltens\fB
			}{
				\vliin{\mixr}{}{
					\vdash \Gamma, \Delta, \fA, \fB
				}{
					\vlpr{\pi_1}{}{\vdash \Gamma, \fA}
				}{
					\vlpr{\pi_2}{}{\vdash \Delta, \fB}
				}
			}
		}
		\\\\
		\vlderivation{
			\vlin{\wdpr}{}{
				\vdash \Gamma_1,\ldots, \Gamma_n, \ncon\connn{\fB_1,\ldots, \fB_n}, \con\connn{\funit, \fA_2, \ldots, \fA_n},\fA
			}{
				\vliiiin{\dconr}{}{
					\vdash \Gamma_1,\ldots, \Gamma_n, \ncon\connn{\fB_1,\ldots, \fB_n}, \con\connn{\fA, \fA_2, \ldots, \fA_n}
				}{
					\vlpr{\pi_1}{}{\vdash \Gamma_1, \fA, \fB_1 }
				}{	
					\vlpr{\pi_2}{}{\vdash \Gamma_2, \fA_2,\fB_2}
				}{\vlhy{\qquad\cdots\qquad}}{
					\vlpr{\pi_n}{}{\vdash \Gamma_n, \fA_n,\fB_n}
				}
			}
		}
		\wdprred
		\vlderivation{
			\vliin{\wdtr}{}{
				\vdash \Gamma_1,\ldots, \Gamma_n,  \ncon\connn{\fB_1,\ldots, \fB_n},
				\con\connn{\funit, \fA_2, \ldots, \fA_n},\fA
			}{
				\vlpr{\pi_1}{}{\vdash \Gamma_1, \fA, \fB_1 }
			}{
				\vlin{2\times\unitor}{}{
					\vdash \Gamma_2,\ldots, \Gamma_n, 
					\ncon\connn{\funit,\fB_1,\ldots, \fB_n},
					\con\connn{\funit, \fA_2, \ldots, \fA_n}
				}{
					\vliiin{\dconr[\fC]}{}{
						\vdash \Gamma_2,\ldots, \Gamma_n, 
						\nfC\connn{\fB_2,\ldots, \fB_n},
						\fC\connn{\fA_2, \ldots, \fA_n}
					}{
						\vlpr{\pi_2}{}{\vdash \Gamma_2, \fB_2,\fC_2}
					}{\vlhy{\qquad\cdots\qquad}}{
						\vlpr{\pi_n}{}{\vdash \Gamma_n, \fB_n,\fC_n}
					}
				}
			}
		}
		\\\\
		\vlderivation{
			\vlin{\wdpr}{}{
				\vdash \Gamma_1,\Gamma_2, \con\connn{\funit,\fB_2, \ldots, \fB_n},\fA
			}{
				\vliin{\wdtr}{}{
					\vdash \Gamma_1,\Gamma_2, \con\connn{\fA,\fB_2, \ldots, \fB_n}
				}{
					\vlpr{\pi_1}{}{\vdash \Gamma_1,\fB_{k}}
				}{
					\vlpr{\pi_2}{}{\vdash \Gamma_2, \con\connn{\fA,\fB_2, \ldots, \fB_{k-1},\funit, \fB_{k+1},\ldots, \fB_n}}
				}
			}
		}
		\wdprred
		\vlderivation{
			\vliin{\wdtr}{}{
				\vdash \Gamma, \con\connn{\funit,\fB_2, \ldots, \fB_k,\fB', \fB_{k+1},\ldots, \fB_n},\fA
			}{
				\vlpr{\pi_1}{}{\vdash \Gamma_1,\fB'}
			}{
				\vlin{\wdpr}{}{
					\vdash \Gamma_2, \con\connn{\funit,\fB_2, \ldots, \fB_{k-1},\funit, \fB_{k+1},\ldots, \fB_n},\fA
				}{
					\vlpr{\pi_2}{}{\vdash \Gamma_2, \con\connn{\fA,\fB_2, \ldots, \fB_{k-1},\funit, \fB_{k+1},\ldots, \fB_n}}
				}
			}
		}
		\\\\
		\vlderivation{
			\vlin{\wdpr}{}{
				\vdash \Gamma, \con\connn{\funit,\fB_2, \ldots, \fB_{n-1},\funit},\fA
			}{
				\vlin{\unitor}{}{
					\vdash \Gamma, \con\connn{\fA,\fB_2, \ldots, \fB_{n-1},\funit}
				}{
					\vlpr{\pi_1}{}{
						\vdash \Gamma, \fC\connn{\fA,\fB_2, \ldots, \fB_{n-1}}
					}
				}
			}
		}
		\wdprred
		\vlderivation{
			\vlin{\unitor}{}{
				\vdash \Gamma, \con\connn{\funit,\fB_2, \ldots, \fB_{n-1},\funit},\fA
			}{
				\vlin{\wdpr}{}{
					\vdash \Gamma, \fC\connn{\funit,\fB_2, \ldots, \fB_{n-1}},\fA
				}{
					\vlpr{\pi_1}{}{
						\vdash \Gamma, \fC\connn{\fA,\fB_2, \ldots, \fB_{n-1}}
					}
				}
			}
		}
	\end{array}$}
	
	\caption{Steps to eliminate $\wdpr$ rules.}
	\label{fig:wdprElim}
\end{figure*}

\begin{lemma}\label{lem:deepr}
	The rule $\deepr$ is admissible in $\PMLx$.
\end{lemma}
\begin{proof}
	Since $\fcC\consfu\neq\funit$, 
	then
	w.l.o.g., 
	$\fcC\conso=\con\connn{\fcC'\conso,\fB'_1, \ldots,\fB'_n}$.
	If $\fcC'\conso=\chole$, then w.l.o.g.,
	$\fB=\fC\connn{\fB'_1, \ldots,\fB'_n}$
	and we conclude since we have 
	$$
	\vlderivation{
		\vliin{\wdpr}{}{
			\vdash \Gamma, \Delta, \con\connn{\fA,\fB'_1, \ldots,\fB'_n}
		}{
			\vlhy{\vdash \Gamma, \fA}
		}{
			\vlin{\unitor}{}{
				\vdash \Delta, \con\connn{\funit,\fB'_1, \ldots,\fB'_n}
			}{
				\vlhy{\vdash \Delta, \fC\connn{\fB'_1, \ldots,\fB'_n}}
			}
		}
	}
	$$
	Otherwise
	we conclude by inductive hypothesis on the size of $\fcC\conso$
	since by \Cref{lem:seqSplit} we can define a derivation of the form 
	$$
	\vlderivation{
		\vlde{\pi_0}{}{
			\vdash \Gamma,\Delta, \con\connn{\fcC\cons{\fA},\fB_1, \ldots,\fB_{k-1},\funit,\fB_{k+1}\ldots \fB_n}
		}{\vlin{\unitor}{}{
				\vdash \Gamma',\Delta',\con\connn{\fcC\cons{\fA},\fB_1, \ldots,\fB_{k-1},\funit,\fB_{k+1}\ldots \fB_n}
			}{
				\vlpr{}{\IH}{\vdash \Gamma',\Delta',\fC\connn{\fcC\cons{\fA},\fB_1, \ldots, \fB_n}}
			}
		}
	}	
	\quor
	\vlderivation{
		\vlde{\pi_0}{}{
			\vdash \Gamma,\Delta, \con\connn{\fcC\cons{\fA},\fB'_1, \ldots,\fB'_n}
		}{\vliiiin{\rho}{}{
				\vdash \Gamma',\Delta',\con\connn{\fcC\cons{\fA},\fB'_1, \ldots,\fB'_n}
			}{
				\vlpr{}{\IH}{\vdash \Gamma', \Delta_0, \fcC\cons{\fA}}
			}{
				\vlpr{\pi_1}{}{\vdash \Delta_1,\fB'_1}
			}{
				\vlhy{\cdots}
			}{
				\vlpr{\pi_n}{}{\vdash \Delta_n,\fB'_n}
			}
		}
	}
	$$
	with $\rho\in\set{\lparr,\ltensr,\dconr}$.
\end{proof}


\subsection{Graph Isomorphism as Logical Equivalence}\label{subsec:sound}

In this sub-section we prove that two formulas $\phi$ and $\psi$ are interpreted by a same graph (i.e., $\gof\phi=\gof\phi$) iff they are logically equivalent (i.e., $\phi\limpeq\psi$).
For this purpose, we show that all equivalence and De Morgan laws from \Cref{def:feq} can be reformulated as logical equivalences.

We first prove that connectives symmetries are derivable in $\PML$.
\begin{lemma}\label{lem:iso}
	The following rules are admissible in $\PML$.
	\begin{equation}
		\vlidf{\symr}{
			\text{\scriptsize$\sigma\in\Sym{\gQ}$}
		}{
			\vdash\Gamma, \con\connn{\fA_{\sigma(1)}, \ldots, \fA_{\sigma(n)}}
		}{
			\vdash\Gamma, \con\connn{\fA_1, \ldots, \fA_n}
		}
		\qquad
		\vlidf{\dsymr}{
			\text{\scriptsize$\rho\in\Dsym{\gQ}$}
		}{
			\vdash\Gamma, \con\connn{\fA_{\rho(1)}, \ldots, \fA_{\rho(n)}}
		}{
			\vdash\Gamma,\cneg\con\connn{\fA_1, \ldots, \fA_n}
		}
	\end{equation}
\end{lemma}
\begin{proof}
	By \Cref{thm:cutelim}, it suffices to prove that the following implications are derivable.
	$$
	\underbrace{
		\begin{array}{l}
			\con\connn{\fA_1, \ldots, \fA_n }\limp \con\connn{\fA_{\sigma(1)}, \ldots, \fA_{\sigma(n)}}	
			\\
			\con\connn{\fA_{\sigma(1)}, \ldots, \fA_{\sigma(n)}} \limp \con\connn{\fA_1, \ldots, \fA_n}
	\end{array}}_{\mbox{ for all }\sigma\in\Sym{\gQ}}
	\qquand
	\underbrace{
		\begin{array}{l}
			\con\connn{\fA_1, \ldots, \fA_n }\limp \ncon\connn{\fA_{\rho(1)}, \ldots, \fA_{\rho(n)}}
			\\
			\ncon\connn{\fA_{\rho(1)}, \ldots, \fA_{\rho(n)}}\limp \con\connn{\fA_1, \ldots, \fA_n}
	\end{array}}_{\mbox{ for all }\tau\in\Dsym{\gQ}}
	$$
	These are easily derivable using an instance of $\dconr$ and $\AXrule$-rules.
\end{proof}

\begin{remark}
	The rule $\symr[\lpar]$ is derivable directly because sequents are sets if occurrences of formulas, therefore the order of the occurrences of the formulas in a sequent is not relevant, and we can permute this order before applying the rule $\lpar$. 
	This because the interpretation of the meta-connective comma we use to separate formulas in a sequent is the same  of $\lpar$.
	
	Similarly, the rule $\symr[\ltens]$ is derivable because in our sequent system, as in standard sequent calculus, the order of the premises of the rules is not relevant.
	Said differently, the space between branches in a derivation is a commutative meta-connective which is internalized by the $\ltensr$.
\end{remark}

Similarly we can prove that the associativity of $\lpar$ and $\ltens$ is derivable.

\begin{lemma}\label{lem:asso}
	The following rules are admissible.
	
	\begin{equation}
		\adjustbox{max width = .9\textwidth}{$
			\vlidf{\passr}{}{
				\vdash\Gamma, \fA_1\lpar(\fA_2\lpar \fA_3)
			}{
				\vdash\Gamma, (\fA_1\lpar \fA_2)\lpar \fA_3
			}
			\qquad
			\vlidf{\tassr}{}{
				\vdash\Gamma, \fA_1\ltens(\fA_2\ltens \fA_3)
			}{
				\vdash\Gamma, (\fA_1\ltens \fA_2)\ltens \fA_3
			}
			$}
	\end{equation}
\end{lemma}
\begin{proof}
	The result follows by \Cref{thm:cutelim} after showing that following implications hold:
	$$\begin{array}{r@{\;\limp\;}l@{\qquad}r@{\;\limp\;}l}
		\fA_1\lpar(\fA_2\lpar \fA_3)
		&
		(\fA_1\lpar \fA_2)\lpar \fA_3
		&
		(\fA_1\lpar \fA_2)\lpar \fA_3
		&
		\fA_1\lpar(\fA_2\lpar \fA_3)
		\\
		\fA_1\ltens(\fA_2\ltens \fA_3)
		&
		(\fA_1\ltens \fA_2)\ltens \fA_3
		&
		(\fA_1\ltens \fA_2)\ltens \fA_3
		&
		\fA_1\ltens(\fA_2\ltens \fA_3)
	\end{array}$$
\end{proof}

We can therefore immediately conclude that $\PML$ is sound and complete with respect to graph isomorphism if we consider unit-free formulas.
\begin{proposition}\label{prop:feq:PML}
	Let $\fA$ and $\fB$ be unit-free formulas.
	Then $\fA\feq\fB$ iff $\proves[\PML]\fA\limpeq \fB$.
\end{proposition}
\begin{proof}
	By induction on the formulas $\fA$ and $\fB$ using \Cref{lem:asso,lem:iso}.
\end{proof}

For a stronger result on general formulas, we need to show that for any two formulas $\fA$ and $\fB$ are interpreted (via $\gof\cdot$) by the same non-empty graph, 
both these formulas are equivalent to a unit-free formula $\fC$ representing the modular decomposition of this graph via graphical connectives.

\begin{theorem}\label{thm:impISsound}
	Let $\fA$ and $\fB$ be pure formulas.
		Then $\fA\feq\fB$ iff $\proves[\PMLx]\fA\limpeq \fB$.
\end{theorem}
\begin{proof}
	Given any formula $\fA$, we can define
	by induction on the number of units $\funit$ occurring in 
	a unit-free formula $\fA'$ such that $\fA\limpeq\fA'$.
	\begin{itemize}
		\item if $\fA$ is a literal, then $\fA'=\fA$;
		\item if $\fA=\con\connn{\fA_1,\ldots, \fA_n}$ and $\fA_i\neq\funit$ for all $i\in\intset1n$, then $\fA'=\con\connn{\fA_1',\ldots, \fA_n'}$.
		Otherwise, w.l.o.g., we assume $\fA_i=\funit$ and we let $\fA'=\fC\connn{\fA_2',\ldots, \fA_n'}$ for a compact formula  $\fC$ such that 
		$\gof{\con\connn{\funit,\fA_2,\ldots,\fA_n}}
		=
		\gof{\fC\connn{\fA_2,\ldots, \fA_n}}$
		and we conclude by inductive hypothesis since we the following derivations:
		$$
		\hskip-2em
		\vlderivation{
			\vlin{\lparr}{}{
				\vdash
				\ncon\connn{\funit,\nfA_2,\ldots,\nfA_n}
				\lpar
				\fC\connn{\fA'_2,\ldots, \fA'_n}
			}{
				\vlin{\unitor}{}{
					\vdash
					\ncon\connn{\funit,\nfA_2,\ldots,\nfA_n}
					,
					\fC\connn{\fA'_2,\ldots, \fA'_n}
				}{
					\vliiin{\dconr[\fC]}{}{
						\vdash
						\nfC\connn{\nfA_2,\ldots,\nfA_n}
						,
						\fC\connn{\fA'_2,\ldots, \fA'_n}
					}{
						\vlpr{\IH}{}{\vdash \nfA_2,\fA'_2}
					}{\vlhy{\cdots}}{
						\vlpr{\IH}{}{\vdash \nfA_n,\fA'_n}
					}
				}
			}
		}
		\quand
		\vlderivation{
			\vlin{\lparr}{}{
				\vdash
				\nfC\connn{\fA'_2,\ldots, \cneg{\fA'}_n}
				\lpar
				\con\connn{\funit,\fA_2,\ldots,\fA_n}
			}{
				\vlin{\unitor}{}{
					\vdash
					\nfC\connn{\cneg{\fA'}_2,\ldots, \cneg{\fA'}_n}
					,
					\con\connn{\funit,\fA_2,\ldots,\fA_n}
				}{
					\vliiin{\dconr[\fC]}{}{
						\vdash
						\nfC\connn{\funit,\cneg{\fA'}_2,\ldots, \cneg{\fA'}_n}
						,
						\fC\connn{\funit,\fA_2,\ldots,\fA_n}
					}{
						\vlpr{\IH}{}{\vdash \cneg{\fA'}_2,\fA_2}
					}{\vlhy{\cdots}}{
						\vlpr{\IH}{}{\vdash \cneg{\fA'}_n,\fA_n}
					}
				}
			}
		}
		$$
		
	\end{itemize}

	Therefore we can find unit-free formulas $\fA'$ and $\fB'$ such that 
	$\fA\limpeq\fA'$ and $\fB\limpeq\fB'$.
	Moreover, by definition of $\gof\cdot$ and the rule $\unitor$  we have $\gof\fA'=\gof\fA=\gof\fB=\gof\fB'$.
	We conclude by \Cref{prop:feq:PML} and the transitivity of $\limpeq$.
\end{proof}

\section{Soundness and Completeness of $\PMLx$ with respect to $\GS$}\label{sec:GS}

In this section we prove that set of graphs which are derivable in the graphical logic $\GS$ from \cite{acc:hor:str:LICS2020,acc:LMCS} 
is the same set of graph corresponding to formulas which are provable in $\PMLx$.

In \Cref{fig:DIrules} we recall the definition of the rules of the deep inference system%
\footnote{
	The definition of deep inference systems operating on graphs can be found in \cite{acc:LMCS} or in  \Cref{app:deep}.
}
$\GS=\Set{\aidr,\sdr,\sur,\pdr}$.

\begin{remark}\label{rem:GS}
	At the syntactical level, the system $\GS$ operates on graphs by manipulating their
	modular decompositions trees.
	Therefore, for any graph occurring in a derivation in $\GS$ we assume a unique formula $\fof\gG$ to be given.
	Note that in $\GS$ the authors allow themselves to consider modular decomposition trees in which leaves may be empty graphs, corrisponding to formulas with unit.
\end{remark}

\begin{figure*}[t!]
	\centering
	\adjustbox{max width=\textwidth}{$
		\begin{array}{c@{\quad}c@{\quad}c@{\quad}c}
			\vlinf{\aidr}{}{\cneg a\lpar a}{\unit}
			&
			\vlinf{\pdr}{}
			{\nP\connn{\gM_1,\ldots,\gM_n}\lpar \gP\connn{\gM'_1,\ldots,\gM'_n}}
			{(M_1\lpar N_1)\ltens \cdots \ltens\, (\gM_n\lpar \gM'_n)}
			\\\\
			\vlinf{\sdr}{}{\gM_i \lpar \gP\connn{\gM_1, \ldots, \gM_{i-1}, \gN,\gM_{i+1}, \ldots, \gM_n}}{\gP\connn{\gM_1, \ldots, \gM_{i-1}, \gM_i\lpar \gN, \gM_{i+1},\ldots \gM_n}}
			&
			\vlinf{\sur}{}{\gP\connn{\gM_1,\ldots, \gM_{i-1},\gM_i \ltens \gN,\gM_{i+1}, \ldots, \gM_n}}{\gM_i\ltens\gP\connn{\gM_1,\ldots,\gM_{i-1}, \gN, \gM_{i+1},\ldots, \gM_n}}
		\end{array}
		$}
	\caption{Inference rules for the system $\GS$, where $\gP$ is a prime graph and $\gM_i\neq\unit\neq\gM_i'$ for all $i\in\intset1n$.}
	\label{fig:DIrules}
\end{figure*}

\begin{remark}
	The set of rules we consider here is a slightly different formulation of with respect to \cite{acc:hor:str:LICS2020} and \cite{acc:LMCS}:
	we consider a $\prule$-rules with a stronger side condition (all factors to be non-empty)
	which is balanced by the presence of $\sur$ in the system.
	The proof that the formulation we consider in this paper is equivalent to the ones in the literature is provided in \Cref{app:GS}.
\end{remark}

We can easily prove that each sequent provable in $\PMLx$ is interpreted by $\gof\cdot$ as a graph which is admitting a proof in $\GS$.

\begin{lemma}[label=lem:GStoSEQ]
	Let $\Gamma$ be a sequent.
	If $\proves[\PMLx]\Gamma$,
	then $\proves[\GS]\gof{\Gamma}$.
\end{lemma}
\begin{proof}
	\def\transto{\;\rightsquigarrow\;}
	We define a derivation $\gof\pi$ of $\gof \Gamma$ in $\GS$ 
	by induction by induction on the last rule $\rrule$ in a derivation $\pi$ of $\Gamma$ in $\PMLx$ 
	according to \Cref{fig:GMLtoGS}.
\end{proof}

\begin{figure}[t]
	\adjustbox{max width = \textwidth}{$\begin{array}{c}
		\vlinf{\axrule}{}{\vdash a,\cneg a}{}
		\transto
		\odiso{\odn{\unit}{\aidr}{a\lpar \cneg a}{}}{\gof{a,\cneg a}}
		\qquad
		\vlderivation{\vlin{\lparr}{}{\vdash \Delta, \fA\lpar \fB}{\vldr{\pi_1}{\vdash \Delta, \fA,\fB}}}
		\transto
		\odiso{
			\od{\odp{\gof{\pi_1}}{\gof{\Delta, \fA,\fB}}{\IH}}
		}{
			\gof{\Delta, \fA\lpar\fB}
		}
		\qquad
		\vlderivation{
			\vliin{\ltensr}{}{
				\vdash\Delta_1, \Delta_2,\fA\ltens\fB
			}{
				\vldr{\pi_1}{\vdash\Delta_1, \fA}
			}{
				\vldr{\pi_n}{\vdash\Delta_2, \fB}
			}	
		}
		\transto
		\odiso{
			\odn{
				\od{\odp{\gof{\pi_1}}{\odiso{\gof{\Delta_1,  \fA}}{\gof{\Delta_1} \lpar \gof{\fA}}}{\IH}}
				\ltens
				\od{\odp{\gof{\pi_2}}{\odiso{\gof{\Delta_2,\fB}}{\gof{\Delta_2} \lpar \gof\fB}}{\IH}}
			}{\pdr}{
				\gof{\Delta_1} \lpar \gof{\Delta_2}
				\lpar 
				\left(\fA\ltens\fB \right)
			}{}
		}{
			\gof{\Delta_1 , \Delta_2, \fA \ltens\fB }
		}
		\\\\
		\vlderivation{\vliiin{\dconr}{}{\vdash \Delta_1,\ldots,\Delta_n, \fP\connn{\fA_1, \ldots, \fA_n},\nfP\connn{\fB_1, \ldots, \fB_n}}
			{\vldr{\pi_1}{\vdash \Delta_1, \fA_1, \fB_1}}{\vlhy{\qquad\dots\qquad}}{\vldr{\pi_n}{\vdash \Delta_n, \fA_n, \fB_n}}
		}
		\quad \transto\quad
		\odiso{\odn{
				\od{\odp{\dD_{\pi_1}}{
						\odiso{\gof{\Delta_1,\fA_1,\fB_1}}{\gof{\Delta_1} \lpar (\gof{\fA_1}\lpar \gof{\fB_1})}}{\IH}}
				\ltens\cdots\ltens
				\od{\odp{\dD_{\pi_n}}{
						\odiso{\gof{\Delta_n,\fA_n,\fB_n}}{\gof{\Delta_1} \lpar (\gof{\fA_n}\lpar \gof{\fB_n})}}{\IH}}
			}{\pdr}{
				\mlpar\connn{\gof{\Delta_1}, \ldots ,\gof{\Delta_n}}
				\lpar
				\odn{
					(\gof{\fA_1}\lpar \gof{\fB_1})\ltens \cdots\ltens (\gof{\fA_n}\lpar \gof{\fB_n})
				}{\pdr}{
					\gP\connn{\gof{\fA_1}, \ldots, \gof{\fA_n}}
					\lpar
					\nP\connn{\gof{\fB_1}, \ldots, \gof{\fB_n}}
				}{}
			}{}
		}{
			\gof{\Delta_1,\ldots, \Delta_n,\fP\connn{\fA_1, \ldots,  \fA_n},\nfP\connn{\fB_1, \ldots, \fB_n}}
		}
		\\\\
		\vlderivation{\vliin{\mixr}{}{\vdash \Delta_1,\Delta_2}{\vldr{\pi_1}{\vdash \Delta_1}}{\vldr{\pi_2}{\vdash \Delta_2}}}
		\transto
		\odiso{
			\ods{\unit}{\gof{\pi_1}}{\gof{\Delta_1}}{\IH}
			\lpar
			\ods{\unit}{\gof{\pi_1}}{\gof{\Delta_2}}{\IH}
		}{\gof{\Delta_1,\Delta_2}}
		\qquad
		\vlinf{\unitor}{}{
			\vdash \Gamma, \con\connn{\fA_1,\ldots, \fA_k,\funit,\fA_{k+1},\ldots,\fA_n}
		}{
			\vdash \Gamma,\fC\connn{\fA_{\sigma(1)},\ldots, \fA_{\sigma(n)}}
		}
		\transto
		\odiso{
			\gof{\vdash \Gamma,\fC\connn{\fA_{\sigma(1)},\ldots, \fA_{\sigma(n)}}}
		}{
			\gof{\Gamma, \con\connn{\fA_1,\ldots, \fA_k,\funit,\fA_{k+1},\ldots,\fA_n}}
		}
		\\\\
		%
		%
		\vlderivation{\vliin{\wdtr}{}{
				\vdash \Delta_1,\Delta_2, \fP\connn{\fA_1, \ldots, \fA_n}
			}{
				\vldr{\pi_1}{\vdash \Delta_1, \fA_1}
			}{
				\vldr{\pi_2}{\vdash \Delta_n, \fC\connn{\fA_2,\ldots, \fA_n}}
			}
		}
		\quad \transto\quad
		\odiso{
			\odn{
				\od{\odp{\dD_{\pi_1}}{\odiso{\gof{\Delta_1,\fA_1}}{\gof{\Delta_1} \lpar \gof{\fA_1} }}{\IH}}
				\ltens
				\od{\odp{\dD_{\pi_2}}{\odiso{\gof{\Delta_2,\fC\connn{\fA_2,\ldots, \fA_n}}}{\gof{\Delta_1} \lpar \gof{\fC\connn{\fA_2,\ldots, \fA_n}}}}{\IH}}
			}{\pdr}{
				\left(\gof{\Delta_1}\lpar \gof{\Delta_2}\right)
				\lpar 
				\odiso{\gof{\fA_1}\ltens \gof{\fC\connn{\fA_2, \ldots, \fA_n}}}{
					\odn{\gof{\fA_1}\ltens \gof{\fC}\connn{\gof{\fA_2}, \ldots, \gof{\fA_n}}}{\sur}{
						\gP\connn{\gof{\fA_1}, \ldots,  \gof{\fA_n}}
					}{}
				}
			}{}
		}{
			\gof{\Delta_1, \Delta_2,\fP\connn{\fA_1, \ldots,  \fA_n}} 
		}
	\end{array}$}
	\caption{Rules to translate derivations in $\PMLx$ into derivations in $\GS$.}
	\label{fig:GMLtoGS}
\end{figure}

To prove the converse, we use the admissibility of $\wdpr$ to prove that 
every time there is a rule in $\GS$ with premise $\gH$ and conclusion $\gG$, 
then there are formulas $\fA$ and $\fB$ such that $\gof\fA$ and $\gof\fB$, and such that $\fB\limp \fA$.
\begin{lemma}\label{lem:DrulesCons}
	Let $\rrule\in\set{\sdr,\sur,\pdr}$.
	If $\vlinf{\rrule}{}{\gG}{\gH}$,
	then there are formulas 
	$\fA$ and $\fB$ with
	$\gof\fA={\gG}$ and $\gof\fB={\gH}$
	such that 
	$\proves[\PMLx]\cneg\fB, \fA$.
\end{lemma}
\begin{proof}
	We first discuss the case if $\cC\conso=\chole$:
	\begin{itemize}	
		\item 
		if $\rrule=\sdr$, then 
		$\fA=\fM_i\lpar \con\connn{\fM_1, \ldots, \fM_{i-1}, \funit\lpar\fN, \fM_{i+1},\ldots \fM_n}$ 
		and 
		$\fB=\con\connn{\fM_1, \ldots, \fM_{i-1}, \fM_i\lpar \fN, \fM_{i+1},\ldots \fM_n}$ 
		for some formulas $\mu_1,\ldots,\mu_n,\nu$ such that $\gof{\fM_i}=\gM_i$ for all $i\in\intset1n$ and $\gof\fN=\gN$.
		We conclude by \Cref{lem:gAX,lem:wdpr} since we have the following derivation
		$$\vlderivation{
			\vlin{\lpar}{}{
				\vdash
				\cneg\fB, \fA
			}{\vlin{\wdpr}{}{
					\vdash
					\cneg{\fB}, \fM_i , \con\connn{\fM_1, \ldots, \fM_{i-1}, \funit \lpar \fN,\fM_{i+1}, \ldots, \fM_n}
				}{
					\vlin{\AXrule}{}{
						\vdash\cneg\fB, \con\connn{\fM_1, \ldots, \fM_{i-1}, \fM_i \lpar \fN,\fM_{i+1}, \ldots, \fM_n}}{\vlhy{}}
			}}
		}
		$$
		
		\item
		if $\rrule=\sur$ then 
		$\fA=\con\connn{\fM_1, \ldots, \fM_{i-1}, \fM_i\ltens\fN, \fM_{i+1},\ldots \fM_n}$ 
		and 
		$\fB= \fM_i\ltens \con\connn{\fM_1, \ldots, \fM_{i-1}, \funit\ltens \fN, \fM_{i+1},\ldots \fM_n}$ 
		for some formulas $\mu_1,\ldots,\mu_n,\nu$ such that $\gof{\fM_i}=\gM_i$ for all $i\in\intset1n$ and $\gof\fN=\gN$.
		We conclude by \Cref{lem:gAX,lem:wdpr} since we have the following derivation
		$$
		\vlderivation{
			\vlin{\lpar}{}{
				\vdash
				\cneg\fB, \fA
			}{\vlin{\cpr}{}{
					\vdash
					\nfM_i, \ncon\connn{\nfM_1, \ldots, \nfM_{i-1},  \funit\lpar\nfN, \nfM_{i+1},\ldots \nfM_n}
					,
					\fA
				}{
					\vlin{\AXrule}{}{
						\vdash
						\ncon\connn{\nfM_1, \ldots, \nfM_{i-1}, \nfM_i\lpar\nfN, \nfM_{i+1},\ldots \nfM_n}
						,
						\fA
					}{\vlhy{}}
			}}
		}
		$$

		\item
		if $\rrule=\pdr$ then 
		$\fA=\nfP\connn{\fM_1,\ldots,\fM_n}\lpar \fP\connn{\fN_1,\ldots,\fN_n}$
		and
		$\cneg\fB=(\cneg\fM_1\ltens \cneg\fN_1)\lpar \cdots \lpar(\cneg\fM_n\ltens \cneg\fN_n)$
		for some formulas 
		$\mu_1,\ldots,\mu_n,\nu_1,\ldots, \nu_n$ such that $\gof{\fM_i}=\gM_i\neq\unit$ and $\gof{\fN_i}=\gN_i\neq\unit$ for all $i\in\intset1n$.
		We conclude since we have the following derivation
		
		$$
		\vlderivation{
			\vliq{\lpar}{}{
				(\cneg\fM_1\ltens \cneg\fN_1)
				\lpar \cdots \lpar
				(\cneg\fM_n\ltens \cneg\fN_n)
				,
				\fA
			}{
				\vliiin{\dconr}{}{
					\vdash 
					(\cneg\fM_1\ltens \cneg\fN_1)
					, \ldots ,
					(\cneg\fM_n\ltens \cneg\fN_n),
					\fA
				}{
					\vliin{\ltens}{}{\vdash\cneg\fM_1\ltens \cneg\fN_1 , \fM_1,\fN_1}{\vlin{\AXrule}{}{\vdash \fM_1 ,\cneg\fM_1}{\vlhy{}}}{\vlin{\AXrule}{}{\vdash \fN_1 ,\cneg\fN_1}{\vlhy{}}}
				}{
					\vlhy{\;\cdots\;}
				}{
					\vliin{\ltens}{}{\vdash\cneg\fM_n\ltens \cneg\fN_n , \fM_n,\fN_n}{\vlin{\AXrule}{}{\vdash \fM_n ,\cneg\fM_n}{\vlhy{}}}{\vlin{\AXrule}{}{\vdash \fN_n ,\cneg\fN_n}{\vlhy{}}}
				}
			}
		}
		$$
		
	\end{itemize}

	If 
	$\cC\conso=\con[\gP]\connn{\cC'\conso, \gM_1, \ldots, \gM_n}\neq\chole$, then we assume w.l.o.g., 
	there is a context formula
	$\fcC\conso=\fP\connn{\fcC'\conso, \fM_1, \ldots, \fM_n}$
	such that
	$\gof{\fcC\conso}=\cC\conso$ and $\gof{\fcC'\conso}=\cC'\conso$ .
	We conclude since, by inductive hypothesis on the structure of $\cC\conso$, there is a derivation of the following form:
	$$
	\vlderivation{
		\vliiiin{\dconr}{}{
			\vdash 
			\nfP\Connn{\cneg{\left(\fcC'\cons{\fB'}\right)}, \nfM_1, \ldots, \nfM_n}
			,
			\fP\Connn{\fcC'\cons{\fA'}, \fM_1, \ldots, \fM_n}
		}{
			\vlpr{}{\IH}{\vdash \cneg{\left(\fcC'\cons{\fB'}\right)} , \fcC'\cons{\fA'}}
		}{
			\vlin{\AXrule}{}{\vdash \cneg\fM_1, \fM_1}{\vlhy{}}
		}{
			\vlhy{\cdots}
		}{
			\vlin{\AXrule}{}{\vdash \cneg\fM_n, \fM_n}{\vlhy{}}
		}
	}
	\qquad.		
	$$
\end{proof}

%
%
We are now able to prove the main result of this section, that is, establishing a correspondence between graphs provable in $\GS$ and graphs which are interpretation via $\gof\cdot$ of formulas provable in $\PMLx$.

%

\begin{theorem}\label{thm:GSisISOMIX}
	Let $\gG\neq\unit$ be a graph and $\fA$ a pure formula such that $\gof\fA=\gG$.
	Then
	$\proves[\GS]\gG$ iff $\proves[\PMLx]\fA$.
\end{theorem}
\begin{proof}
	By  \Cref{lem:GStoSEQ}, if $\proves[\PMLx]\fA$, then by there is a proof of $\gof\fA$ in $\PMLx$.
	
	To prove the converse, let $\dD$ be a proof of $\gG\neq\unit$ in $\GS$.
	We define a proof $\pi_\dD$ of $\fA$ 
	by induction on the number $n$ of rules in $\dD$.
	
	\begin{itemize}
		\item 
		We cannot have $n=0$ since we are assuming $\gG\neq\unit$.
		
		\item 
		If $n=1$, then $\gG=a\lpar \cneg a$
		and
		$\pi_\dD=\vlsmash{\vlderivation{\vlin{\lpar}{}{\vdash a\lpar \cneg a}{\vlin{\axrule}{}{\vdash a, \cneg a}{\vlhy{}}}}}$
		.
		
		\item 
		If $n>1$, then $\dD=\odn{\od{\odp{\dD'}{\gH}{}}}{\rrule}{\gG}{}$
		then by inductive hypothesis
		we have a proof $\pi_{\dD'}$ of a formula $\fB$ such that
		$\gof\fB=\gH$.
		If $\rrule\in\set{\sdr,\sur,\pdr}$, then by  \Cref{lem:DrulesCons}
		we have a derivation with $\cutr$ as the one below on the left of a formula $\fA$ such that $\gof\fA=\gG$.
		We then conclude by \Cref{thm:cutelim}.
		$$\begin{array}{c|c}
		\vlderivation{
			\vliin{\cutr}{}{\vdash \fA}{\vlpr{}{\IH}{\psi}}{\vlpr{}{\text{\Cref{lem:DrulesCons}}}{\vdash \cneg\fB, \fA}}	
		}
		\quad\overset{\Cref{thm:cutelim}}{\rightsquigarrow^*}\quad
		\vlderivation{\vlpr{}{\PMLx}{\fA}}
		\qquad\quad&\quad\qquad
		\vlderivation{
			\vlid{=}{}{
				\vdash \fA
			}{
				\vliin{\deepr}{}{
					\vdash \fcC\cons{a\lpar\cneg a}
				}{
					\vlin{\lparr}{}{\vdash a\lpar \cneg a}{\vlin{\axrule}{}{\vdash a, \cneg a}{\vlhy{}}} 
				}{
					\vlpr{\pi_{\dD'}}{\IH}{\fB}
				}	
			}
		}
		\end{array}$$
		Otherwise $\rrule=\aidr$, then it must have been applied deep inside a context 
		$\cC\conso=\gof{\fcC\conso}\neq\chole$ such that $\cC\consu=\gH=\gof{\fB}$.
		Therefore $\fA=\fcC\cons{a\lpar\cneg a}$.
		We conclude by applying 
		\Cref{lem:deepr} to the derivation above on the right.
		\qedhere
	\end{itemize}
\end{proof}

\section{Classical Logic Beyond Cographs}\label{sec:LK}

We conclude this paper by providing an extension of $\PML$ with standard \emph{contraction} and \emph{weakening} structural rules, showing that it provides a conservative extension of propositional classical logic.
We then show decomposition results allowing us to factorize any proof into a linear proof (i.e., a proof in $\PML$) and a \emph{resource management} proof (i.e., a derivation only using weakening and contraction rules).

\begin{definition}
	We define the following sequent system:
	\begin{equation}\label{eq:LKsys}
		\begin{array}{l@{\colon\;}l@{\;=\;}l}
			\mbox{\defn{Classical Graphical  Logic}}
			&
			\PLK
			&
			\PML\cup\set{\wrule,\crule}
		\end{array}
	\end{equation}	
\end{definition}

\begin{figure*}
	\centering
	\adjustbox{max width=\textwidth}{$\begin{array}{c}
			\begin{array}{c@{\qquad}|@{\qquad}c@{\qquad}|@{\qquad}c}
				\vlinf{\wrule}{}{\vdash \Gamma, \fA}{\vdash\Gamma}
				\qquad
				\vlinf{\crule}{}{\vdash \Gamma, \fA}{\vdash\Gamma, \fA, \fA}
				&
				\vlinf{\dwrule}{}{\fB\lpar\fA}{\fB}
				\qquad
				\vlinf{\dcrule}{}{\fA}{\fA\lpar \fA}
				&
				\vlinf{\dacrule}{}{a}{a \lpar a}
				\qquad
				\vlinf{\medr}{\text{\scriptsize$\lpar\neq\gP$ prime}}{\gP\connn{\fA_1\lpar \fB_1,\ldots, \fA_n\lpar \fB_n}}{
					\gP\connn{\fA_1,\ldots, \fA_n}\lpar	\gP\connn{\fB_1,\ldots, \fB_n}
				}
			\end{array}
		\end{array}$}
	\caption{Structural rules for sequent calculi, and the corresponding rules in deep inference together with the atomic contraction and the generalized medial rule.}
	\label{fig:WCrules}
\end{figure*}

For $\PLK$ we can prove the admissibility of the $\cutr$-rule via cut-elimination.

\begin{theorem}[Cut-elimination]\label{thm:cutelimLK}
	The rule $\cutr$ is admissible in $\PLK$.
\end{theorem}
\begin{proof}
	Consider the \emph{cut-elimination steps} from \Cref{fig:cut-elimPML} and \Cref{fig:cut-elimWC}
	and the definition of weight from the proof of \Cref{thm:cutelim}.
	A proof of weak normalization of the cut-elimination procedure can be given using the same measure used in the proof of \Cref{thm:cutelim}
	\footnote{
		For the sake of determining if a cut-formula is principal, 
		in a contraction rule ($\crule$) we assume both occurrences of $\fA$ in the premise to be active and the occurrence of $\fA$ in the conclusion to be principal.
	}
	by restraining the application of the cut-elimination steps only to top-most $\cutr$-rules in the derivation.
\end{proof}

\begin{figure*}
	\centering
	\adjustbox{max width=\textwidth}{$\begin{array}{c}
		\vlderivation{
			\vliin{\cutr}{}{\vdash \Gamma, \Delta}{
				\vlin{\wrule}{}{\vdash \Gamma,\fA}{
					\vlhy{\vdash \Gamma}}
			}{
				\vlhy{\vdash \nfA, \Delta}
			}
		}
		\;\rightsquigarrow\;
		\vlderivation{\vliq{\wrule}{}{\vdash \Gamma,\Delta}{
				\vlhy{\vdash \Gamma}}}
		\qquad\qquad
		\vlderivation{
			\vliin{\cutr}{}{\vdash \Gamma, \Delta}{
				\vlin{\crule}{}{\vdash \Gamma, \fA}{
					\vlhy{\vdash \Gamma, \fA, \fA}}
			}{
				\vlhy{\vdash  \nfA , \Delta}
			}
		}
		\;\rightsquigarrow\;
		\vlderivation{
			\vliq{\crule}{}{\vdash \Gamma, \Delta}{
				\vliin{\cutr}{}{\vdash \Gamma, \Delta, \Delta}{
					\vliin{\cutr}{}{\vdash \Gamma, \Delta, \fA}{
						\vlhy{\vdash \Gamma, \fA, \fA}
					}{
						\vlhy{\vdash  \nfA , \Delta}
					}
				}{
					\vlhy{\vdash  \nfA , \Delta}
				}
			}
		}
	\end{array}$}
	\caption{The cut-elimination steps for the structural rules.}
	\label{fig:cut-elimWC}
\end{figure*}

We consider the deep inference rules in \Cref{fig:WCrules}, that is, rules which can be applied on subformulas in a sequent.
Using 
the deep inference version of the structural rules (weakening and contraction) and the
\defn{generalized medial} rule proposed in \cite{CDW:ext-bool}
we can define a inference system 
where structural rules can pushed down in a derivation
obtaining a decomposition result extending the one in
\cite{brunnler:06:locality,bru:str:swi-med} for classical logic.

\begin{lemma}\label{lem:dCdecomposition}
	The contraction rule $\dcrule$ is derivable using atomic contraction ($\dacrule$) and medial rule ($\medr$).
\end{lemma}
\begin{proof}
	By induction on the contracted formula $\fA$.
	If $\fA=a$ is an atom, then $\dcrule$ is an instance of $\dacrule$.
	Otherwise, $\fA=\con\connn{\fB_1,\ldots, \fB_n}$ and
	we conclude since we can apply inductive hypothesis 
	to replace each application of $\dcrule$ with a derivation of the following form
	$$
	\odn{\con\connn{\fB_1,\ldots, \fB_n}\lpar \con\connn{\fB_1,\ldots, \fB_n}}{\dcrule}{
		\con\Connn{\fB_1,\ldots, \fB_n}
	}{}
	\qquad\rightsquigarrow\qquad
	\odn{\con\connn{\fB_1,\ldots, \fB_n}\lpar \con\connn{\fB_1,\ldots, \fB_n}}{\medr}{
		\con\Connn{\ods{\fB_1\lpar \fB_1}{\IH}{\fB_1}{\set{\medr,\dacrule}},\ldots, \ods{\fB_n\lpar \fB_n}{\IH}{\fB_n}{\set{\medr,\dacrule}}}
	}{}
	\; .
	$$
\end{proof}

\def\dWC{\set{\dwrule,\dcrule}}
\begin{theorem}[Decomposition]\label{thm:GLKdecomposition}
	Let $\Gamma$ be a sequent.
	If $\proves[\PLK]\Gamma$,
	then:
	\begin{enumerate}
		\item\label{dec:1}
		there is a sequent $\Gamma'$ such that 
		$\proves[\PML]\Gamma'
		\proves[\set{\dwrule,\dcrule}]\Gamma$
		
		\item\label{dec:2}
		there are sequent $\Gamma'$, $\Delta'$, and $\Delta$
		such that 
		$\proves[\PML]\Gamma'
		\proves[\set{\medr}]\Delta'
		\proves[\set{\dacrule}]\Delta
		\proves[\set{\dwrule}]\Gamma$
	\end{enumerate}
\end{theorem}
\begin{proof}
	The proof of \Cref{dec:1} is immediate by applying rule permutations. 
	For a reference, see \cite{acc:str:18}.
	\Cref{dec:2} is consequence of the previous point 
	since by \Cref{lem:dCdecomposition} we can replace all instances of $\dcrule$-rules with derivations containing only $\medr$ and $\dacrule$, and conclude by applying rule permutations to move all $\acrule$-rules below the instances of $\medr$-rules, and $\dwrule$ to the bottom of a derivation.
\end{proof}

To conclude this section, we recall that classical graphical logic, is not the same logic of the \defn{boolean graphical logic} (denoted $\GBL$) defined in \cite{CDW:ext-bool} (an inference systems on graphs by extending the semantics of read-once boolean relations from cographs to general graphs).
In fact, even if both are conservative extensions of classical logic, 
the following graph from \cite{acc:LMCS} which is expected to be provable in $\GBL$, 
but is not provable in $\GS$ (and there is no formula $\fA$ provable in $\PLK$ such that $\gof\fA$ is the given graph).
$$
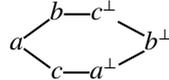
\begin{array}{cc@{\quad}cc}
	&\vb1& \vnc1&		\\
	\va1&	 &		&\vnb1	\\
	&\vc1&\vna1
\end{array}
\Gedges{a1/b1,b1/nc1,nc1/nb1,nb1/na1,na1/c1,a1/c1}
$$

\section{Conclusion and Future Works}\label{sec:conc}

In this paper we have provided foundations for the design of proof systems operating on graph by defining \emph{graphical connectives}, 
a class of logical operators generalizing the classical conjunction and disjunction,
and 
whose semantics is solely defined by their interpretation as prime graphs.

We studied two sub-structural sequent calculi operating on formulas defined via graphical connectives ($\PML$ and $\PMLx$), proving that cut-elimination holds in these systems and that they are conservative extensions of 
the multiplicative linear logic and the multiplicative linear logic with mix respectively.
For these calculi, we proved that they capture graph isomorphisms as provable logical equivalences%
\footnote{
	Note that the sequent calculus is only capable of checking if two graphs sharing the same set of vertices are isomorphic (problem with polynomial complexity), but not to find an correspondence between vertices of two graphs which is an isomorphism (a well-known $\NP$ problem)
}%
.
We were able to prove that the class of graphs representing provable formulas in $\PMLx$
coincides with 
the class of non-empty graphs provable in the proof system $\GS$ from \cite{acc:hor:str:LICS2020,acc:hor:str:2021}.
As a consequence, the proofs of cut-elimination in $\PML$ serves as simplified version of the proof of transitivity of implication in $\GS$.

We concluded by providing a conservative extension of both classical propositional logic and $\PML$, 
and proving the existence of a decomposition result allowing us to have canonical forms for proofs in which all structural rules can be relegated at the bottom of a derivation.

\subsection{Future Works}

\myparagraph{Categorical Semantics}
Unit-free \emph{star-autonomous} and \emph{IsoMix} categories \cite{cockett:seely:97:mix,cockett:seely:97} provide categorical models of $\MLL$ and $\MLLx$ respectively.
We conjecture that categorical models for $\PML$ and $\PMLx$ can be defined by enriching such structures with additional $n$-ary monoidal products and natural transformation, reflecting the symmetries observed in the symmetry groups of prime graphs.

\myparagraph{Digraphs, Games and Event Structures}
In this work we have extended the correspondence between classical propositional and cographs from \cite{cographs} to the case of general (undirected) graphs using graphical connectives. The same idea can be found in \cite{acc:FSCD22} for mixed graphs.
A similar generalization of the correspondence between intuitionistic propositional formulas and \emph{arenas} used in Hyland-Ong \emph{game semantics} \cite{hyland:ong:00}. Arenas are directed graphs characterized by the absence of specific induced subgraphs.
We foresee the possibility of defining conservative extensions of intuitionistic propositional logic beyond arenas, analogously as what done in this paper, where we provided conservative extensions of classical propositional logic beyond cographs.
Such proof systems would provide new insights on the proof theory connected to concurrent games~\cite{abramsky1999concurrent,rideau2011concurrent,winskel2017games}, and could be used to define automated tools operating on event structures~\cite{nielsen1981petri}.

\begin{figure}[t]
	\newvertex{nbtensa}{\cneg b\lpar \cneg a}{}
	\newvertex{btensa}{a \ltens b}{}
	$$
	\begin{array}{c|c}
		\begin{array}{c@{\qquad}c}
			\\
			\begin{array}{cc}
				\vnb1\;\; \vna1\;\;\; & \;\;\; \va1 \quad \vb1
				\\
				\Gpar1&\Gtens1
				\\[5pt]
				\vnbtensa1
				&
				\vbtensa1
			\end{array}
			\pswires{Gpar1.O/nbtensa1,Gtens1.O/btensa1}
			\psaxioms{Gpar1.R/Gtens1.L,Gpar1.L/Gtens1.R}
			&
			\begin{array}{c@{\;}c@{\;}cc@{\;}c@{\;}c}
				\vnb1 &&  \vna1 & \va1 && \vb1
				\\
				&\vvo{\lpar}1&&& \vvo{\ltens}2
				\\[5pt]
				&\vvr{\lpar}1&&& \vvr{\ltens}2
			\end{array}
			\Bedges{vo1/vr1,vo2/vr2}
			\Redges{a1/b1}
			\multiRedges{a1,b1}{vo2}
			\multiRedges{na1,nb1}{vo1}
			\sqBedges{a1/na1/10,nb1/b1/15}
		\end{array}
		&
		\begin{array}{cc}
			\\
			\begin{array}{c}
				\vnc1 \quad \vna1	\quad\vnd1	\quad\vnb1
				\\[5pt]
				\vvo{\nPfour}2
				\\[5pt]
				\vvr{\nPfour}2
			\end{array}
			&
			\begin{array}{c}
				\va1\quad \vb1 \quad \vc1 \quad \vd1	
				\\[5pt]
				\vvo{\Pfour}1
				\\[5pt]
				\vvr{\Pfour}1
			\end{array}
		\end{array}
		\Hpath{
			(na1.center)--
			++(0,20pt)	-|
			(a1.center)	--
			(b1.center)	--
			++(0,16pt)	-|
			(nb1.center)--
			(nd1.center)--
			++(0,8pt)	-|
			(d1.center)	--
			(c1.center)	--
			++(0,12pt)	-|
			(nc1.center)--
			(na1.center)
		}
		\Redges{a1/b1,b1/c1,c1/d1}
		\Redges{nc1/na1,nd1/nb1,na1/nd1}
		\multiRedges{vo1}{a1,b1,c1,d1}
		\multiRedges{vo2}{na1,nb1,nc1,nd1}
		\sqBedges{a1/na1/20,b1/nb1/16,c1/nc1/12,d1/nd1/8}
		\Bedges{vo1/vr1,vo2/vr2}
	\end{array}
	$$
	\caption{
		On the left: the same proof net in the original Girard's syntax and Retoré's one.
		On the right: an \RBPN of $\fcon{\Pfour}\connn{a,b,c,d}\limp \fcon{\Pfour}\connn{a,b,c,d}$ containing the chorded \ae-cycle $a\cdot b \cdot \cneg b\cdot \cneg d \cdot d \cdot c \cdot \cneg c\cdot \cneg a$.
	}
	\label{fig:mllnets}
\end{figure}

\myparagraph{Proof nets and automated proof search}
We plan to design proof nets~\cite{girard:87,dan:reg:89,girard:96:PN} for $\PML$ and $\PMLx$, as well as combinatorial proofs~\cite{hughes:invar} for $\PLK$.
For this purpose, we envisage to extend Retoré's \emph{handsome proof net} syntax, where proof nets are represented by two-colored graphs (see the left of \Cref{fig:mllnets}). In Retoré's syntax, the graph induced by the vertices corresponding to the inputs of a $\lpar$-gate (or a $\ltens$-gate) is similar to the corresponding prime graph $\lpar$ (resp.~$\ltens$). Thus, gates for graphical connectives could be easily defined by extending this correspondence (see the proof net on the right of \Cref{fig:mllnets}).
The standard correctness condition defined via \emph{acyclicity} would fails for these proof nets, as shown in the right-hand side of \Cref{fig:mllnets}: the (correct) proof-net of the sequent $\Pfour\connn{a,b,c,d}\limp \Pfour\connn{a,b,c,d}$ contains a cycle.
We foresee the possibility of using results on the \emph{primeval} decomposition of graphs \cite{jam:ola:pcomp,hougardy1996p4} to isolate those cycles witnessing unsoundness, as proposed in \cite{seiller-CohGraphs}.
Such a result would open to the possibility of defining \emph{combinatorial proofs} \cite{hughes:invar,hughes:pws} for $\PLK$ relying on the decomposition result (\Cref{thm:GLKdecomposition}), and may provide a methodology to find machine-learning guided automated theorem provers for $\PML$ and $\PMLx$ using the methods in \cite{kog:moo:moo:neural}.

\clearpage
\bibliography{GL}

\clearpage

\appendix

\section{Deep Inference and the Open Deduction Formalism}\label{app:deep}

Open deduction~\cite{gug:gun:par:2010} is a proof formalism based on deep inference~\cite{tub:str:esslli19}. It has originally been defined for formulas, but it is abstract enough such that it can equally well be used for graphs, as already done in~\cite{acc:hor:str:2021}.
%

\begin{definition}\label{def:derivation}
	An \defn{inference system} $\sysS$ is a set of inference rules (as for example shown in \Cref{fig:rules}).
	A \defn{derivation} $\dD$ in $\sysS$ with premise $\gG$ and
	conclusion $\gH$ is denoted $\vldownsmash{\od{\odd{\odh{\gG}}{\dD}{\gH}{\sysS}}}$
	and is defined inductively as follows:
	\
	
	\begin{itemize}	
		\item 
		Every graph $\gG$ is a (\defn{trivial}) derivation
		with
		premise $\gG$ and conclusion $\gG$ (also denoted $\gG$).

		\item 
		An instance of a rule $\vlinf{\rrule}{}{\gH}{\gG}$  in $ \sysS$ is a derivation with premise $\gG$ and conclusion $\gH$.
		
		\item 
		If 
		$\dD_1$ is a derivation with premise $\gG_1$ and conclusion $\gH_1$, 
		and 
		$\dD_2$ is a derivation with premise $\gG_2$ and conclusion $\gH_2$,
		and
		$\gH_1\greq\gG_2$,
		then the composition 
		of $\dD_1$ and $\dD_2$ 
		is a derivation $\dD_2\comp\dD_1$ denoted as below.
		\begin{equation*}\label{eq:iso}
			\odiso{
				\ods{\gG_1}{\dD_1}{\gH_1}{\sysS}
			}{
				\ods{\gG_2}{\dD_2}{\gH_2}{\sysS}
			}
			\mbox{ or }
			\ods{\odiso{\ods{\gG_1}{\dD_1}{\gH_1}{\sysS}}{\gG_2}}{\dD_2}{\gH_2}{\sysS}
			\mbox{ or }
			\ods{\gG_1}{\dD_1}{\odiso{\gH_1}{\ods{\gG_2}{\dD_2}{\gH_2}{\sysS}}}{\sysS}
			\mbox{ or }
			\od{\odd{\odd{\odh{\gG_1}}
					{\dD_1}{\gG_2}{\sysS}}
				{\dD_2}{\gH_2}{\sysS}}
			\mbox{ or }	
			\od{\odd{\odd{\odh{\gG_1}}
					{\dD_1}{\gH_1}{\sysS}}
				{\dD_2}{\gH_2}{\sysS}}
		\end{equation*}
		Note that even if the symmetry between $\gG_2$ and $\gH_1$ is not written, we always assume it is part of the derivation and explicitly
		given.
		
		\item   
		If $\gG$ is a graph with $n$ vertices and
		$\dD_1,\dots,\dD_n$ are derivations with premise $\gG_i$ and conclusion
		$\gH_i$\; for each $i\in\intset1n$,
		then $\gG\connn{\dD_1,\dots, \dD_n}$ is a derivation with
		premise $\gG\connn{\gG_1,\dots, \gG_n}$
		and conclusion $\gG\connn{\gH_1,\dots, \gH_n}$
		denoted as below on the left.
		$$
		\begin{array}{c|c}
			\od{\odh{
					\gG\Connn{
						\od{\odd{\odh{\gG_1}}{\dD_1}{\gH_1}{\sysS}}
						,\dots,
						\od{\odd{\odh{\gG_n}}{\dD_n}{\gH_n}{\sysS}}
					}
				}
			}
			\quad&\quad
			\od{\odh{\ods{\gG_1}{\dD_1}{\gH_1}{}\star \ods{\gG_1}{\dD_1}{\gH_1}{}}}	
		\end{array}
		$$
		If $\gG=\star\in\set{\lpar,\ltens}$ we may write the derivations as above on the right.
	\end{itemize}
	Therefore, 
	$
	\cC\Cons{\vldownsmash{\od{\odd{\odh{\gG}}{\dD}{\gH}{\sysS}}}}
	\coloneqq
	{\od{\odd{\odh{\cC\cons\gG}}{\cC\cons\dD}{\cC\cons\gH}{\sysS}}}
	$
	is well-defined
	for any 
	context $\cC\conso$ and any derivation
	${\od{\odd{\odh{\gG}}{\dD}{\gH}{\sysS}}}$.

	A \defn{proof} in $\sysS$ is a derivation in $\sysS$ whose premise
	is~$\unit$. 
	
	A graph $\gG$ is \defn{provable} in $\sysS$ (denoted $\proves[\sysS]{\gG}$) iff there is a proof in $\sysS$ with conclusion $\gG$.
\end{definition}

\subsection{Equivalent Definitions of $\GS$}\label{app:GS}

\def\wpdr{\mathsf{p_1}\mathord{\downarrow}}
\def\wwpdr{\mathsf{p_2}\mathord{\downarrow}}

We here show that the formulation of the system $\GS$ provided in this paper is equivalent to one provided in \cite{acc:hor:str:LICS2020,acc:LMCS}.
In particular, in the previous these papers the rule $\sur$ was not included in the system.
However, as shown in \cite{acc:LMCS} this rule plays a crucial role in the proof that $\GS$ is a conservative extension of $\MLLx$
and in 
\cite{acc:FSCD22} it is shown that this rule cannot be admissible in the proof systems operating on mixed graphs.
Moreover, we here give a weaker side condition on the $\prule$-rule with respect to the rules below:

\begin{equation}
	\adjustbox{max width=.9\textwidth}{$\begin{array}{c|c}
		\pdr \mbox{ in \cite{acc:LMCS}} 
		& 
		\pdr \mbox{ in \cite{acc:hor:str:LICS2020}}
		\\\hline\\
		\vlinf{\wpdr}{\star}{\nP\connn{\gM_1,\ldots, \gM_n}\lpar \gP\connn{\gN_1,\ldots, \gN_n}}{(\gM_1\lpar \gN_1)\ltens\cdots\ltens(\gM_n\lpar\gN_n)}
		&
		\vlinf{\wwpdr}{\dagger}{\nP\connn{\gM_1,\ldots, \gM_n}\lpar \gP\connn{\gN_1,\ldots, \gN_n}}{(\gM_1\lpar \gN_1)\ltens\cdots\ltens(\gM_n\lpar\gN_n)}
		\\\\
		\star\coloneqq\begin{array}{l}
			\gP\notin\set{\lpar,\ltens} \mbox{ prime }
			\mbox{$\gM_i\neq\unit$ for all $i\in\intset1n$}
		\end{array}
		&
		\dagger\coloneqq\begin{array}{l}
			\gP\notin\set{\lpar,\ltens} \mbox{ prime }
			\mbox{$\gM_i\lpar\gN_i\neq\unit$ for all $i\in\intset1n$}
		\end{array}
	\end{array}$}
\end{equation}

In order to prove the equivalence between our system and the ones in \cite{acc:hor:str:LICS2020,acc:LMCS} we recall the following lemma allowing us to prove that in $\GS$ we can derive any graph of the shape $\gG\limp\gG$.

\begin{lemma}\label{lem:g}
	Let $\gM_1, \dots, \gM_n,\gN_1, \dots,\gN_n$ and $\gG$ be graphs 
	such that 
	$\sizeof{\vG}=n$.
	Then there is a derivation
	$$
	\od{
		\odd{\odh{(\gM_1 \lpar \gN_1)\ltens \cdots\ltens(\gM_n\lpar \gN_n)}}{}
		{\cneg\gG\connn{\gM_1, \ldots, \gM_n} \lpar \gG\connn{\gN_1, \ldots, \gN_n}}{\set{\sur,\pdr}}
	}
	$$
\end{lemma}
\begin{proof}
	By induction on the modular decomposition of $\gG$.
\end{proof}

Thanks to this lemma, we can therefore prove the admissibility of the weaker

\begin{proposition}\label{prop:wpdr}
	The following version of $\pdr$ with weaker side conditions
	is admissible in $\GS$
	$$
	\vlinf{\wpdr}{
		\text{\scriptsize$\gP$ prime, $\gM_i\neq\unit$ for all $i\in\intset1n$}
	}{\nP\connn{\gM_1,\ldots, \gM_n}\lpar \gP\connn{\gN_1,\ldots, \gN_n}}{(\gM_1\lpar \gN_1)\ltens\cdots\ltens(\gM_n\lpar\gN_n)}
	$$
\end{proposition}
\begin{proof}
	Note that we may have $\gN_i=\unit$ for some $i\in\intset1n$.
	Thus, if $\gN_i\neq\unit$ for all $i\in\intset1n$, then $\wpdr$ is an occurrence of $\pdr$.
	Otherwise, w.l.o.g.,
	$\gN_1=\unit$, thus
	we have a derivation
	$$
	\odiso{
		\gM_1\ltens
		\od{\odd{\odh{(\gM_2\lpar \gN_2)\ltens\cdots\ltens(\gM_n\lpar\gN_n)}}{}{
				\nH\connn{\gM_2,\ldots, \gM_n}
				\lpar
				\gH\connn{\gN_2,\ldots, \gN_n}
			}{\text{\Cref{lem:g}}}}
	}{
		\odn{\gM_1\ltens\nP\connn{\unit, \gM_2,\ldots, \gM_n}}{\sur}{\nP\connn{\gM_1, \gM_2,\ldots, \gM_n}}{} \lpar \gP\connn{\unit, \gN_2,\ldots, \gN_n}
	}
	$$
\end{proof}

\begin{theorem}
	Let $\gG$ be a graph.
	Then 
	$$
	\proves[\GS]\gG
	\Leftrightarrow
	\proves[\set{\aidr,\sdr,\sur,\wpdr}]\gG
	\Leftrightarrow
	\proves[\set{\aidr,\sdr,\wpdr}]\gG
	\Leftrightarrow
	\proves[\set{\aidr,\sdr,\wwpdr}]\gG
	$$
\end{theorem}
\begin{proof}
	The first equivalence follows from \Cref{prop:wpdr}.
	The other has been proved in \cite{acc:LMCS}.
\end{proof}
\section{On Rules Introducing a Connective at a Time}\label{app:sconr}
\def\MLLC{\MLL^{\conr}}

A rule introducing only one connective (different from $\lpar$ and $\ltens$) at a time inevitably leads to the same problem observed in the literature of \emph{generalized multiplicative connectives} \cite{dan:reg:89,girard2000meaning,mai:19,acc:mai:20}, 
where \emph{initial coherence} (i.e. the possibility of having only atomic axioms in a cut-free system, \cite{avr:canonical:01}) is ruled out because of the so-called \emph{packaging problem}.

However, in this appendix we discuss the results about the system extending multiplicative linear logic with the rule $\conr$, that is, the system. 
$$
\MLLC \coloneqq\Set{\axrule,\lparr,\ltensr, \mixr,\conr[\con[\gP]] \mid\gP\in\primesign}
\quad\mbox{where}\quad
\vliiinf{\conr[\con[\gP]]}{}{
	\vdash\Gamma_1, \ldots, \Gamma_n, \con[\gP]\connn{\fA_1, \ldots, \fA_n}
}{
	\vdash\Gamma_1, \fA_1
}{
	\qquad\cdots\qquad
}{
	\vdash\Gamma_n, \fA_n
}	
$$

We first observe that in the system does not satisfy anymore initial coherence;
e.g., the formula $\con[\Pfour]\connn{a,b,c,d}\limp \con[\Pfour]\connn{a,b,c,d}$ is not provable anymore.
However, the system still satisfies cut-elimination. The proof cut-elimination is straightforward by considering the following additional cut-elimination steps.

\begin{equation*}
\adjustbox{max width =\textwidth}{$
	\begin{array}{c}
		\vlderivation{
			\vliin{\cutr}{}{
				\vdash\Gamma_1,\ldots, \Gamma_n, \Delta_1, \ldots, \Delta_n
			}{
				\vliiin{\conr}{}{
					\vdash\Gamma_1,\ldots, \Gamma_n,\fP\connn{\fA_1, \ldots, \fA_n}
				}{\vlhy{\vdash\Gamma_1,\fA_1}}{\vlhy{\qquad\cdots\qquad}}{\vlhy{\vdash\Gamma_n,\fA_n}}
			}{
				\vliiin{\conr}{}{
					\vdash\Delta_1,\ldots, \Delta_n, \nfP\connn{\cneg\fA_1, \ldots, \cneg\fA_n}
				}{\vlhy{\vdash\Delta_1,\cneg\fA_1}}{\vlhy{\qquad\cdots\qquad}}{\vlhy{\vdash\Delta_n,\cneg\fA_n}}
			}
		}
		\;\rightsquigarrow\;
		\vlderivation{
			\vliiiq{\mixr}{}{
				\vdash\Gamma_1,\ldots, \Gamma_n, \Delta_1, \ldots, \Delta_n
			}{
				\vliin{\cutr}{}{\vdash\Gamma_1,\Delta_1}{\vlhy{\vdash\Gamma_1,\fA_1}}{\vlhy{\vdash\Delta_1,\cneg\fA_1}}
			}{
				\vlhy{\quad\cdots\quad}
			}{
				\vliin{\cutr}{}{\vdash\Gamma_n,\Delta_n}{\vlhy{\vdash\Gamma_n,\cneg\fA_n}}{\vlhy{\vdash\Delta_n,\cneg\fA_n}}
			}
		}
		\\\\
		\vlderivation{
			\vliin{\cutr}{}{
				\vdash\Gamma_1,\ldots, \Gamma_n, \Delta_1, \ldots, \Delta_n,\nfP\connn{\fB_1, \ldots, \fB_n}
			}{
				\vliiin{\dconr}{}{
					\vdash\Gamma_1,\ldots, \Gamma_n,\nfP\connn{\fB_1, \ldots, \fB_n},\fP\connn{\fA_1, \ldots, \fA_n}
				}{\vlhy{\vdash\Gamma_1,\fA_1,\fB_1}}{\vlhy{\qquad\cdots\qquad}}{\vlhy{\vdash\Gamma_n,\fA_n,\fB_n}}
			}{
				\vliiin{\conr}{}{
					\vdash\Delta_1,\ldots, \Delta_n, \nfP\connn{\cneg\fA_1, \ldots, \cneg\fA_n}
				}{\vlhy{\vdash\Delta_1,\cneg\fA_1}}{\vlhy{\qquad\cdots\qquad}}{\vlhy{\vdash\Delta_n,\cneg\fA_n}}
			}
		}
		\;\rightsquigarrow\;
		\vlderivation{
			\vliiin{\conr}{}{
				\vdash\Gamma_1,\ldots, \Gamma_n, \Delta_1, \ldots, \Delta_n,\ncon\connn{\fB_1,\ldots,\fB_n}
			}{
				\vliin{\cutr}{}{\vdash\Gamma_1,\Delta_1,\fB_1}{\vlhy{\vdash\Gamma_1,\fA_1,\fB_1}}{\vlhy{\vdash\Delta_1,\cneg\fA_1}}
			}{
				\vlhy{\quad\cdots\quad}
			}{
				\vliin{\cutr}{}{\vdash\Gamma_n,\Delta_n,\fB_n}{\vlhy{\vdash\Gamma_n,\cneg\fA_n,\fB_n}}{\vlhy{\vdash\Delta_n,\cneg\fA_n}}
			}
		}
	\end{array}
$}
\end{equation*}

Note that $\conr$ is derivable in $\PMLx$.

\begin{lemma}
	The rule $\conr$ is derivable in $\PMLx$.
\end{lemma}
\begin{proof}
	If $\con=\lpar$, then $\conr$ is derivable using $\lparr$ and mix.
	If $\con=\ltens$, then $\conr=\ltensr$.
	Otherwise, we conclude by induction on the arity of $\con$ since we have a derivation
	$$
	\vlderivation{
		\vliin{\wdtr}{}{
			\vdash\Gamma_1, \ldots, \Gamma_n, \con\connn{\fA_1,\ldots, \fA_n}
		}{
			\vlhy{\vdash \Gamma_1, \fA_1}
		}{
			\vlin{\unitor}{}{
				\vdash\Gamma_2, \ldots, \Gamma_n, \con\connn{\funit,\fA_2,\ldots, \fA_n}
			}{
				\vlpr{\dD'}{\IH}{\vdash\Gamma_2, \ldots, \Gamma_n, \fcC\connn{\fA_1,\ldots, \fA_n}}
			}
		}
%
	}
	$$
	where $\dD'$ contains instances of $\conr$ introducing connectives whose arities are strictly smaller then the arity of $\con$.
\end{proof}

\end{document}